\documentclass[12pt]{article}
\usepackage{eurosym}
\usepackage{amsmath}
\usepackage{amssymb}
\usepackage{amsfonts}
\usepackage{graphicx}
\usepackage[utf8]{inputenc}
\usepackage[english]{babel}
\usepackage[utf8]{inputenc}
\usepackage[numbers]{natbib}
\usepackage{enumerate}

\input{tcilatex}

\begin{document}
\title{A String Model of Liquidity in Financial Markets }
% \thanks{With the collaboration of David German, Thanh Hoang, and Jennifer
% Thompson.}
\author{Sergey Lototsky\footnote{Department of Mathematics, University of Southern California.} \and Henry Schellhorn\footnote{Institute of Mathematical Sciences, Claremont Graduate University.} \and Ran Zhao\footnote{Institute of Mathematical Sciences and Drucker School of Management, Claremont Graduate University.}}
% \thanks{University of Southern California}
% \author{}
% \thanks{Claremont Graduate University}
% \author{}
% \thanks{Claremont Graduate University}
\date{}
\maketitle

\begin{abstract}
We consider a dynamic market model of liquidity where unmatched buy and sell limit orders are stored in order books. The resulting net demand surface constitutes the sole input to the model. We prove that generically there is no arbitrage in the model when the driving noise is a stochastic string. Under the equivalent martingale measure, the clearing price is a martingale, and options can be priced under the no-arbitrage hypothesis. We consider several parameterized versions of the model, and show some advantages of specifying the demand curve as quantity as a function of price (as opposed to price as a function of quantity). We calibrate our model to real order book data, compute option prices by Monte Carlo simulation, and compare the results to observed data.

\begin{flushleft}

%\textbf{PACS: } 62F10 $\cdot$ 62F12 $\cdot$ 62M86
\end{flushleft}
\end{abstract}

% \begin{keywords}
%   Liquidity Modeling, String Model, It\^o-Wentzell formula, No-arbitrage Condition, SPDE
% \end{keywords}

% \begin{AMS}
%   91B26, 91G80, 91G60
% \end{AMS}

\date{}

\section{Introduction}

In our model, equilibrium prices of assets are completely determined by the
order flow, which is viewed as an exogenous process. We model a market of
assets without specialist where every trader submits limit orders, that is,
for a buy order, the buyer specifies the maximum purchase price, or buy
limit price, and, for a sell order, the seller specifies the minimum sale
price, or sell limit price\footnote{%
There is no loss of generality in that statement. A buy market order can be
specified in our model as a buy limit order with limit price equal to
infinity. Since we model assets with only positive prices, a sell market
order can be specified in our model as a sell limit order with a limit price
equal to zero.}.

If, at a given moment in time, the buyer is unable to complete the entire
order due to the shortage of sell orders at the required limit price, the
unmatched part of the order is recorded in the order book. A symmetric
outcome exists in the case of incoming sell orders. Subsequently these buy
unmatched orders may be matched with new incoming sell orders. We note that
this is the operating procedure for many electronic exchanges, such as NYSE
Arca. Time-priority is used to break indeterminacies of a match between an
incoming buyer at limit price superior to the ask price, i.e., the lowest
limit price in the sell order book. As a result, the equilibrium of \textit{%
clearing price} process is always defined.

Since the matching mechanism does not add any information to the economy,
all information about asset prices is included in the order flow. Whether
public exchanges should or should not reveal the order book data in real
time is an important issue, which continues to preoccupy the financial
markets community~\cite{WW02}. Our theoretical framework accommodates either
viewpoint. However, our model is most useful in an economy where order books
are public information, but where the large trader position is not known.
The current blossoming of high-frequency trading activity~\cite%
{BRZ,BLT06,Eng00,Hau08} seems to confirm our viewpoint that traders (i) are
interested in understanding order book information, and (ii) trade on that
information.

We do not address in this paper the issue of differential information. The
market microstructure literature, such as the Kyle model~\cite{Kyl85},
considers various setting involving multiple uninformed, or noise, traders,
and one or several informed traders. A key result of ~\cite{Kyl85} is that,
given the information available to noise traders, the resulting price
process is a martingale with respect to a suitable measure, whereas it may
not be for the informed traders. As a consequence, we do not believe that
abstracting issues of differential information results in any loss of
generality. The order books reflect all the public information. Public
information corresponds to the filtration under which the clearing price
needs to have an equivalent martingale measure, in order to avoid arbitrage.
Obviously, the clearing price may not be a martingale in the aforementioned
measure if the filtration is enlarged to include private information.

There are two main classes of models in the liquidity literature. The first
class of models (\cite{BB04, BK15a, BK15b, Fre98, Jar92, Jar94, PS98a,
PS98b, RS10, SW00}) considers the action of a large trader who can
manipulate the prices in the market; our model belongs to this class. The
large trader can employ one of the following two strategies: to ``corner the
market and squeeze the shorts'' or to ``front-run one's own trades''. While
some exchanges have rules to curtail the cornering of the market,
front-running seems more difficult to ban from an exchange. In discrete-time
trading, it is known that absence of arbitrage in periods where the large
investor does not trade implies non-existence of a market manipulation
strategy \cite{BK15a}.

The second class of models, considered, for example, in \cite{CR07, CJP04,
CST10,GS11}, abstracts the issues of market manipulation away, and considers
all traders as price-takers. In particular, \cite{CJP04} introduces an
exogenous residual supply curve against which an investor trades. The
investor trades market orders, the orders are matched instantaneously, and
then, as pointed out in \cite{BB04}, it is reasonable to assume that the
price effect of an order is limited to the very moment when the order is
placed in the market, so that the residual supply curve at a future time is
statistically independent from the order just matched.

The paper by Roch~\cite{RA11} attempts to bridge the gap between these two
classes of models by analyzing a \emph{linear} impact of the large trader on
the demand price. By contrast, our model is not limited to a linear impact.

Our model extends previous models in two directions. In our model, all the
information is contained in a Brownian sheet driving the dynamics of the
market net demand curve. Such a string model makes it possible to represent
any correlation structure in the net demand curve and has already been
introduced in finance to model the yield curve. Santa Clara and Sornette~%
\cite{SCS01} argue that a ``discontinuous with respect to term] forward
curve [...] is intuitively unlikely''. All the uncertainty in the economy is
contained in this single Brownian sheet. We can thus assume that claims can
be replicated by trading at various points on the net demand curve, and thus
the market is complete.

Secondly, our demand curve represents quantity (number of shares) as a
function of price, whereas, traditionally, the demand curve represents the
price as a function of quantity. The main advantage of this formulation is
that we can easily prove a risk-neutral pricing formula in a market that the
large trader can manipulate, but where she limits herself to continuous
strategies of bounded variation in order to avoid liquidity costs. There is
also a technical advantage to this approach (see remark \ref{remmodelPisbad}%
). Interestingly, our risk-neutral model of the demand curve (expressed as
quantity as a function of price) leads to a nonlinear stochastic partial
differential equation (SPDE). A linearization of this SPDE leads to an
unstable model. Could it be the mathematical explanation why illiquid
markets tend to be instable \cite{O04}? We leave this interesting question
for future research.

We proceed in two steps to define our model. In the first step, we consider
a market with atomistic traders and a large trader, and develop conditions
on the net demand curve of the atomistic traders such that the large trader
cannot manipulate the market and, as a result, will refrain from trading
large orders, or orders of infinite variation, that generate liquidity
costs. We assume then in a second step that the large trader is reduced to
differentiable strategies. Generally, large trader strategies are not
observable \cite{OW13}. Fortunately, the risk-neutral pricing formula that
we obtain, under the assumptions of completeness mentioned above, is the
same for every large trader strategy. Thus, even if the large trader can
manipulate the underlying market, she cannot manipulate option prices beyond
the immediate impact of her trade on the current price of the underlying.
This is an expedient feature of our model: we do not need to identify
whether there is a large trader, or several large traders on the market for
the option pricing formula to be plausible.

Our contributions are thus two-fold. First, we show that, under natural
assumptions, the large trader cannot generate arbitrage, and a risk-neutral
pricing formula obtains. Second, we match a particular model to empirical
data. In this particular model, liquidity effects have an impact on option
prices. Similar to most other models that came after Black-Scholes~\cite%
{BS73}, we obtain a smile curve for implied volatility. Our analysis shows
that a fairly technical (infinite-dimensional) model can be easily
implemented. We hope that, should this paper become of interest among
practitioners, more refined implementations will follow than the one
contained in this paper, which is meant only to prove our concept.

The structure of the paper is as follows. Section \ref{sec::pre} covers preliminaries on the market mechanism and string model. Section \ref{sec:MarketModel} introduces the infinite-factor string model and Section \ref{Sec:AFFM} illustrate the possibility of arbitrage in a finite-factor model. Section \ref{sec:NA} derives general conditions for existence of an equivalent martingale measure. Section \ref{sec:PMC} introduces an option pricing formula in a market where a large trader can manipulate the price of the underlying. Section \ref{sec:MOD} introduces the discretized model that we will use in our empirical analysis. Empirical analysis is in Section \ref{sec:NumericalEx}, where we describe the data set, the discretized algorithm, and pricing of options. We provide technical details of proofs in
Section \ref{sec::App} (Appendix). 

\section{Preliminaries}   \label{sec::pre}

\subsection{The Market Mechanism}

\noindent A buy\ limit order specifies how many shares a trader wants to
buy, and the maximum purchase price per share; we call this price the (buy)
\textit{limit price}. A sell\ limit order specifies how many shares a trader
wants to sell, and at what minimum price he is willing to sell them. We call
this price the (sell)\ \textit{limit price}. The unmatched buy and sell
orders are stored in order books, until they are either canceled or matched
with an incoming order. An incoming order is matched with the order on the
opposite side of the market which has the best price. The clearing price of
the transaction is equal to the limit price of the order in the book, and
not of the incoming order. Partial execution is allowed, and ties are
resolved by time-priority. Here is example of the matching mechanism in
discrete time, i.e., at most one order arrives at time $t\in
\{0,1,2,\ldots\}.$

%\textbf{Example}

\begin{example}
Suppose that the clearing price at time $0$ is any price $P(0)\in \lbrack
100,120]$. After clearing, that is, when $0<t<1$ we suppose that the order
books contain the following orders

\begin{equation*}
\begin{tabular}{|l|l|}
\hline
\multicolumn{2}{|l|}{Buy Order Book} \\
Price & Quantity \\ \hline\hline
100 & 10 \\ \hline
\end{tabular}%
\text{ \ \ \ \ \ }%
\begin{tabular}{|l|l|}
\hline
\multicolumn{2}{|l|}{Sell Order Book} \\
Price & Quantity \\ \hline\hline
120 & 10 \\ \hline
130 & 10 \\ \hline
\end{tabular}%
\end{equation*}

At time $t=1$ a buy order arrives with a limit price of \$125, and a
quantity of 15. The exchange matches it with the best sell order, i.e., the
one with a sell limit price of \$120. However, execution is only partial,
and the remainder of the buy order is placed in the order book at the limit
price of \$125, resulting in the following order books:

\begin{equation*}
\begin{tabular}{|l|l|}
\hline
\multicolumn{2}{|l|}{Buy Order Book} \\
Price & Quantity \\ \hline\hline
100 & 10 \\ \hline
125 & 5 \\ \hline
\end{tabular}%
\text{ \ \ \ \ \ }%
\begin{tabular}{|l|l|}
\hline
\multicolumn{2}{|l|}{Sell Order Book} \\
Price & Quantity \\ \hline\hline
130 & 10 \\ \hline
\end{tabular}%
\end{equation*}

The clearing price at time 1 is equal to the limit price of the sell order:
\begin{equation*}
P(1)=120.
\end{equation*}
\end{example}

This example illustrates several properties of the limit order markets:

\begin{itemize}
\item The clearing price is always defined, and can assume any positive value%
\footnote{%
We do not consider markets for swaps, where the price can be negative.}.

\item An incoming order can ``cross'' the order book, i.e., for the case of
a buy order, that its limit price is higher than the best sell order limit
price (i.e, the best ask), since the buyer does not lose a cent. Crossing
the book is indeed advantageous for two reasons: first, it allows for faster
execution. In our example, had the buyer submitted an order at price \$130
he would have bought the complete quantity of shares (15)\ that he desired,
rather than waiting an indeterminate amount of time until enough sell orders
arrive at his limit price.

\item If several buy orders are submitted at the same time and demand
exceeds supply at the best ask, then the buy orders with the highest limit
price are executed first. Our own data analysis shows that few orders cross
the Arca book. This is consistent with the theory of optimal order book
placement suggested by Rosu~\cite{Ros09}.
\end{itemize}

\subsection{String Modeling}

We now move to continuous time, omitting the technicalities related to
convergence of a discrete time model. We start with a filtered probability
space $(\Omega ,\mathcal{F},\{\mathcal{F}_{t}\},\mathbb{P})$ satisfying the
usual conditions. Equalities of random variables are to be understood almost
surely unless stated otherwise. Likewise, after explaining that there exists
a modification of a stochastic process that satisfies a certain property, we
do not distinguish in the remainder of the text between the original process
and the corresponding modification.

All the uncertainty is described by a one-dimensional Brownian sheet $%
B=B(s,t)$ for $0\leq s\leq 1$ and $0\leq t\leq T\ $ which generates $\{%
\mathcal{F}_{t}\equiv \sigma \big(B(s,t)\big);0\leq s\leq 1\}$ on $0\leq
t\leq T$. There are three main approaches to constructing a Brownian sheet
and the corresponding stochastic integral; cf. Mueller~\cite{Mue09}:

\begin{enumerate}
\item the martingale measure approach (Walsh \cite[Chapter 2]{Wal86}),

\item the Hilbert space approach (Da Prato and Zabczyk~\cite[Chapter 4]%
{DPZ2014}),

\item the function space approach (Krylov \cite[Section 8.2]{Kryl99}).
\end{enumerate}

The Hilbert space approach covers the martingale measure approach \cite%
{DPZ2014,Mue09}; the function space approach covers the Hilbert space
approach \cite{Kryl99}.

To construct the stochastic integral with respect to the Brownian sheet
using the function space approach, we take an orthonormal basis $\{\mathfrak{%
m}_{n};n\geq 1\}$ in $L^{2}[0,1]$ and let $\{w_{n},\ n\geq 1\}$ be
independent standard Brownian motions on $[0,T]$. Define the random field
\begin{equation*}
B(s,t)=\sum_{n=1}^{\infty }w_{n}(t)\int_{0}^{s}\mathfrak{m}_{n}(r)dr, \ \
s\in [0,1], \ t\in [0,T].
\end{equation*}

It follows that $B=B(s,t)$ is a Gaussian random field with mean zero and
covariance
\begin{equation*}
\mathbb{E}B(s,t)B(r,u)=\min(r,s)\cdot\min(t,u),
\end{equation*}
and then, by the Kolmogorov continuity criterion, $B$ has a modification
that is jointly continuous in $(s,t)$ \cite[Proposition 1.4]{Wal86}; this
modification is called Brownian sheet. If $b=b(s,t)$ is a random field such
that
\begin{equation}  \label{sq-int-SVL}
\int_0^T\int_{0}^{1}\mathbb{E}b^{2}(s,t)\,dsdt<\infty,
\end{equation}
and, for each $s\in [0,1]$, the process $b(s,\cdot)$ is $\mathcal{F}_{t}-$%
adapted, then, by definition,
\begin{equation}  \label{continuous}
\int_{0}^{t}\int_{0}^{s}b(r,u)B(dr,du)= \sum_{n\geq 1} \int_0^t\left(
\int_0^{s} b(r,u)\mathfrak{m}_n(r)dr\right) dw_n(u).
\end{equation}

Then Girsanov's Theorem (cf. \cite[Theorem 2.2]{Allouba} or \cite[Theorem
10.14]{DPZ2014}) can be stated as follows.

\begin{theorem}
\label{th:Girsanov} Suppose that $\lambda (s,\cdot)$ is an $\mathcal{F}_{t}$%
-predictable process for each $s\in [ 0,1]$, and that
\begin{equation}  \label{E1}
\mathbb{E}\left[\exp \left(\int_{0}^{T}\int_{0}^{1}\lambda (s,u)B(ds,du) -%
\frac{1}{2}\int_{0}^{T}\int_{0}^{1}\lambda ^{2}(s,u)dsdu\right)\right]=1.
\end{equation}
Define a new probability measure $\mathbb{Q}$ on $(\Omega ,\mathcal{F}_T)$
by
\begin{equation}  \label{MM}
d\mathbb{Q}= \exp \left(-\int_{0}^{T}\int_{0}^{1}\lambda(s,u)B(ds,du)- \frac{%
1}{2}\int_{0}^{T}\int_{0}^{1}\lambda^{2}(s,u)ds\right) d\mathbb{P}.
\end{equation}

Let
\begin{equation} \label{eqn::risk_lambda}
B^{\mathbb{Q}}(s,t)=B(s,t)+\int_{0}^{t}\int_0^s\lambda (r,u)drdu,
\end{equation}

Then, for each $s\in [0,1]$, the process $B^{\mathbb{Q}}(s,\cdot)$ is a
standard Brownian motion with respect to $\{\mathcal{F}_{t}\}_{0\leq t\leq T}
$ on the probability space $(\Omega ,\mathcal{F},\mathbb{Q)}$.
\end{theorem}

\begin{corollary}
\label{cor-Girsanov-diffusion} If $X=X(t),\ t\in [0,T], $ is a stochastic
process with representation
\begin{equation*}
X(t)=X(0)+\int_0^t\int_0^{s}\sigma_X(r,u)\lambda(r,u)dsdu+\int_0^t\int_0^{s}%
\sigma_X(r,u)B(dr,du),
\end{equation*}
then $X$ is a martingale under the measure $\mathbb{Q}$.
\end{corollary}

\subsection{It\^{o}-Wentzell Formula}

Let $F=F(x,t),\ x\in \mathbb{R},\ t\in [0,T], $ be a random field and let $%
g=g(t),\ t\in [0,T],$ be a stochastic process such that
\begin{align*}
&F(x,t) =F(x,0)+\int_0^t\mu _{F}(x,u)du+ \int_0^t\int_{0}^{s}
\sigma_{F}(x,r,u)B(dr,du), \\
&g(t) =g(0)+\int_0^t\mu_{g}(u)du+\int_0^t\int_{0}^{s}\sigma_{g}(r,u)B(dr,du).
\end{align*}

\begin{definition}
\label{def:IWC} We say that the pair $(F,g)$ satisfies the It\^{o}-Wentzell
conditions if

\begin{enumerate}
\item The random variables $F(x,0)$, $x\in \mathbb{R}$ and $g(0)$ are $%
\mathcal{F}_0$-measurable;

\item Each of the processes $\mu_g(\cdot)$, $\sigma_g(r,\cdot)$, $%
\mu_F(x,\cdot),\ x\in \mathbb{R},$ and $\sigma_F(x,r,\cdot)$ is $\mathcal{F}%
_t$-adapted;

\item The functions $F$ and $g$ are continuous in $t$.

\item The function $F$ is twice continuously differentiable in $x$ and the
function $\sigma_F$ is continuously differentiable in $x$.

\item The following integrability condition holds:
\begin{align}  \label{Int-IW}
\mathbb{E}\mathcal{I}&<\infty, \mathrm{\ where\ } \\
\mathcal{I}&= \int_{0}^{T}\Big|\mu_{F}\big(g(u),u\big)\Big|du
+\int_{0}^{T}\int_{0}^{1}\sigma_{F}^2\big(g(u),s,u\big)dsdu  \notag \\
& +\int_{0}^{T}\Big|\frac{\partial F}{\partial x}\big(g(u),u\big)\, \mu
_{g}(u)\Big|du +\int_{0}^{T}\int_{0}^{1}\Big|\frac{\partial F}{\partial x}%
\big(g(u),u\big)\, \sigma_{g}(s,u)\Big|^2dsdu  \notag \\
& +\int_{0}^{T}\int_{0}^{1} \Big|\frac{\partial ^{2}F}{\partial x^{2}}\big(%
g(u),u\big)\Big|\, \sigma _{g}^{2}(s,u)dsdu \notag \\
& +\int_{0}^{T}\int_{0}^{1} \Big|%
\frac{\partial \sigma_{F}}{\partial x}\big(g(u),s,u\big)\,\sigma_{g}(s,u)%
\Big|dsdu.  \notag
\end{align}
\end{enumerate}
\end{definition}

\begin{theorem}
If the pair $(F,g)$ satisfies the It\^{o}-Wentzell conditions, then
\begin{equation}
\begin{split}  \label{IWK}
F\big(g(t),t\big)-F\big(g(0),0\big) &= \int_{0}^{t}\mu_{F}\big(g(u),u\big)du
+\int_{0}^{t}\int_{0}^{s}\sigma_{F}\big(g(u),r,u\big)B(dr,du) \\
& +\int_{0}^{t}\frac{\partial F}{\partial x}\big(g(u),u\big)\, \mu _{g}(u)du
\\
& +\int_{0}^{t}\int_{0}^{s}\frac{\partial F}{\partial x}\big(g(u),u\big)\,
\sigma_{g}(r,u)B(dr,du) \\
& +\frac{1}{2}\int_{0}^{t}\int_{0}^{s} \frac{\partial ^{2}F}{\partial x^{2}}%
\big(g(u),u\big)\, \sigma _{g}^{2}(r,u)drdu \\
& +\int_{0}^{t}\int_{0}^{s} \frac{\partial \sigma_{F}}{\partial x}\big(%
g(u),r,u\big)\,\sigma_{g}(r,u)drdu.
\end{split}%
\end{equation}
\end{theorem}

\begin{proof}
This follows by combining \cite[Theorem 3.1]{Kryl11} with \eqref{continuous}%
. Alternatively, one can derive \eqref{IWK} from \cite[Theorem 3.3.1]{Kun97}
by writing
\begin{equation*}
F(x,t)=F(x,0)+\int_0^t \mu_F(x,u)du+M_F(x,t),\ g(t)=g(0)+\int_0^t
\mu_g(u)du+M_g(t),
\end{equation*}
and noticing that
\begin{align*}
\langle M_g\rangle(t)&=\int_0^t\int_0^{s} \sigma^2_g(r,u)dr du, \\
\left\langle \frac{\partial{F}}{\partial x}, g \right\rangle (t)&=
\left\langle \frac{\partial{M_F}}{\partial x}, M_g \right\rangle(t)=
\int_{0}^{t}\int_{0}^{s} \frac{\partial \sigma_{F}}{\partial x}\big(g(u),r,u%
\big)\,\sigma_{g}(r,u)dsdu.
\end{align*}
Note that some form of \eqref{Int-IW} is necessary to define the right-hand
side of \eqref{IWK}.
\end{proof}

\section{A Market with Atomistic Traders and a Large Trader}

\label{sec:MarketModel}

We view atomistic traders as a continuum of traders. Each atomistic trader
submits an infinitesimally small order at a price $p$, and there is no
concentration of orders at any particular price. This lack of order
concentration distinguishes atomistic traders from the large trader.

Similar to \cite[Assumption A1]{Jar92}, we assume that the market is
frictionless, that is, there are no transaction costs. We also assume that
the buy and sell limit price $p$ can take every value between $0$ and $S$,
and the orders can be submitted to the market at any time $t\in [0,T]$.

\begin{definition}
The net demand curve $Q=Q(p,t,\omega)$ is a real-valued function on $%
[0,S]\times [0,T]\times \Omega$. The number $Q(p,t,\omega )$\ is equal to
the difference between the total quantity of shares \textbf{submitted} for
purchase at price lower than or equal to $p$ and the total quantity of
shares \textbf{submitted} for sale at price larger than or equal to $p$
between time $0$ and time $t$ by atomistic traders.
\end{definition}

\textbf{Assumption Q1.} \emph{For every $t\in [0, T]$, the function $%
Q=Q(p,t) $ is twice continuously differentiable and strictly decreasing in
its first argument $p$.}

If $N$ is the total number of shares outstanding on the market, then, for
all $p\in (0,S)$ and $t\in [0,T]$,
\begin{equation}  \label{cross}
-N\leq Q(S,t)<Q(p,t)<Q(0,t)\leq N.
\end{equation}
The fact that the net demand curve is decreasing as a function of the price $%
p$ is an immediate consequence of the market mechanism. Indeed, the number
of shares of available buy orders is decreasing with price, while the number
of shares of available sell orders is increasing: buy low, sell high. Our
Assumption \textbf{Q1} is similar to Assumption A3 in \cite{Jar92} or
Assumption 2 in \cite{BB04}; see also Remark \ref{remarkPQ} below.

Note that, for a fixed price $p$, the total quantity of shares submitted,
either for purchase or for sale, does not need to be increasing in time. As
a result, there is no monotonicity condition  on $Q(p,t)$ as a function of $t
$.

Denote by $P(x,t)$ the price available on the market at time $t\in [0,T]$ when the large trader's position is $x\in [x_{\min},x_{\max}]$, with negative values of $x$ corresponding to a short position. We assume the following \emph{market clearing condition}:
\begin{equation}  \label{MCE-const}
Q\big(P(x,t),t\big)+x = C,
\end{equation}
where $C$ is a constant.

\begin{remark}
\label{remarkPQ}\textrm{Equality \eqref{MCE-const} means that the sum of the
number of shares held by the large trader and the current net demand at the
corresponding price does not depend on time. From a mathematical point of
view, the particular value of the constant on the right hand side of %
\eqref{MCE-const} is not important at this point, and we will set it to
zero:
\begin{equation}
Q\big(P(x,t),t\big)+x=0.  \label{market_clearing_equation}
\end{equation}
By Assumption \textbf{Q1}, $Q(p,t)$ is monotonically decreasing in $p$, and then %
\eqref{market_clearing_equation} implies that $P(x,t)$ is monotonically
increasing in $x$, which is exactly Assumption 2 from \cite{BB04}. By %
\eqref{cross}, $P(x,t)\in [0,S]$ for all $x\in [Q(S,t),Q(0,t)]$, $t\in [0,T]$%
.}
\end{remark}

The position of the large trader at time $t$ is the predictable process $%
\theta=\theta (t)$. It is also called the \textit{large trader trading
strategy}. The process $\theta$ must be a semimartingale satisfying
\begin{equation}  \label{thetabound}
x_{\min} \leq \theta(t)\leq x_{\max}.
\end{equation}

In order for \eqref{market_clearing_equation} to always have a solution, we
assume without loss of generality that:
\begin{equation}  \label{thetaQ}
Q(S,t)<x_{\min}<x_{\max}<Q(0,t).
\end{equation}

As shown in \cite{BB04}, optimal trading strategies are continuous, so, with
no loss of generality, we assume that the process $\theta$ is continuous.
Denote by $\Theta$ the set of all trading strategies, that is, continuous
semimartingales satisfying \eqref{thetabound}.

Next, define the \textit{asymptotic liquidation proceeds} of the large
trader in a fixed position $\vartheta$ at time $t$ by
\begin{equation}  \label{ALP}
L(\vartheta ,t)=\int_{0}^{\vartheta }P(x,t)dx.
\end{equation}

The large trader \textit{holdings in the bank account} are denoted by $%
\beta^{\theta }(t)$. The \textit{\ realizable wealth} of the large trader
achieved by the trading strategy $\theta $ is denoted by $V^{\theta }(t)$,
where
\begin{equation*}
V^{\theta }(t)=\beta^{\theta }(t)+L\big(\theta(t),t\big).
\end{equation*}
In what follows, we use the notation
\begin{equation*}
L\big(\theta(t),dt\big)=\big(L\big(x,t\big)
dP (x,t )\big)\big|_{x=\theta(t)}.
\end{equation*}

\begin{proposition}
\label{prop:RW} For every $\theta\in \Theta, $
\begin{equation}
V^{\theta }(t)-V^{\theta }(0) =\int_0^tL\big(\theta(u),du\big) -\frac{1}{2}%
\int_{0}^{t}\frac{\partial P}{\partial x} \big(\theta(u),u\big)d\langle
\theta \rangle(u),  \label{lemma32BB04}
\end{equation}
where $\langle \theta \rangle$ is the quadratic variation of $\theta$.
\end{proposition}

\begin{proof}
This follows by the It\^{o}-Wentzell formula; see \cite[Lemma 3.2]{BB04} for
details.
\end{proof}

\begin{corollary}
\label{cor:SPM} If the process
\begin{equation*}
t\mapsto \int_{0}^{t}L\big(\theta(u),du\big),\ \ t\in [0,T],
\end{equation*}
is a local martingale under an equivalent martingale measure $\mathbb{Q}$,
then the realizable wealth $V^{\theta}$ is a supermartingale under $\mathbb{Q%
}$.
\end{corollary}

\begin{proof}
Indeed,
\begin{equation*}
\int_{0}^{t}\frac{\partial P}{\partial x} \big(\theta(u),u\big)d\langle
\theta \rangle(u)\geq 0,
\end{equation*}
for all $t\geq 0$, because the process $\langle \theta \rangle $ is
increasing and, by \eqref{market_clearing_equation},
\begin{equation*}
\partial P(x,t)/\partial x>0.
\end{equation*}
\end{proof}

\begin{definition}
An\textit{\ arbitrage strategy} is a trading strategy $\theta\in \Theta$
such that
\begin{equation*}
V^{\theta }(0)=0,\ \ \mathbb{P}\big(V^{\theta }(T) \geq 0\big)=1, \ \
\mathbb{P}\big(V^{\theta }(T) >0\big)>0.
\end{equation*}
A \textit{market model admits arbitrage} if there exists an arbitrage
strategy.
\end{definition}

\section{Arbitrage in Finite-Factor Models}

\label{Sec:AFFM}

Consider the price process \footnote{%
According to Martin Schweizer (private communication), this model is more of
theoretical interest: in most applications, prices are not limited above.
See our online supplement for a more complicated model where prices are not
limited above but that need at least two stocks to generate arbitrage.}
\begin{equation}  \label{eq:P}
P(x,t)=\mu(t)+\sigma(x)h\big(Z(t)\big), \ x\in [x_{\min},x_{\max}],\ t\geq 0,
\end{equation}
where

\begin{itemize}
\item $x_{\min}<0$ is the largest short position the large trader can take;

\item $\mu=\mu(t)$ is a non-random, positive, continuously differentiable
function with $\mu^{\prime }(t)>0$;

\item $\sigma=\sigma(x)$ is a non-random, bounded, strictly increasing
smooth function;

\item $h=h(y),\ y\in \mathbb{R},$ is a strictly positive, bounded function;

\item $Z=Z(t)$ is a real-valued noise process.
\end{itemize}

To ensure that the price stays positive, we assume that
\begin{equation}  \label{eq:pos}
\delta_0= \mu(0) + \min\limits_{x,y} \big(\sigma(x)h(y)\big) \geq 0.
\end{equation}

As a concrete example, take
\begin{equation}  \label{eq:B}
P(x,t)= 20+2t+(2x-1)\big(2+\sin\big(Z(t)\big)\big),\ t\geq 0,\ x\in [-2,2];\
\end{equation}
\begin{equation*}
Z(t)=\sigma_0(t)\int_0^t\int_0^1b_0(s,u)B(ds,du),
\end{equation*}
with a Brownian sheet $B=B(s,t)$ and suitable functions $\sigma_0$ and $b_0$.

If $\sigma(x^*)=0$ for some $x^*>0$, then the  constant (time-independent)
buy strategy
\begin{equation*}
\theta(t)=x^*,\ t\geq 0,
\end{equation*}
is an arbitrage strategy.

Indeed, the asymptotic liquidation process is
\begin{equation*}
\begin{split}
L(x^*,t)&=\int_0^{x^*} P(x,t)dx =\mu(t)x^*+ h\big(Z(t)\big)\int_0^{x^*}
\sigma(x)dx \\
& = x^*\left( \mu(t) +\frac{h\big(Z(t)\big)}{x^*} \int_0^{x^*}
\sigma(x)dx\right).
\end{split}%
\end{equation*}
Strict monotonicity of $\sigma$ means
\begin{equation*}
\frac{1}{x^*} \int_0^{x^*} \sigma(x)dx > \sigma(0),
\end{equation*}
and therefore, by \eqref{eq:pos},
\begin{equation}  \label{eq:L}
L(x^*,t)> x^*\big(\mu(t)+\sigma(0)h\big(Z(t)\big)\big) >x^*\,\delta_0
\end{equation}
for all $t\geq 0$; the second inequality is also strict because $%
\sigma(0)>\sigma(x_{\mathrm{min}})$ and $h\big(Z(t)\big)>0$. As a result,
\begin{eqnarray}  \label{eq:arb}
V^{x^*}(t)-V^{x^*}(0) &=& \int_0^tL(x^*,ds) \notag \\
&=& \int_0^t L(x^{*},s)\mu'(s)ds > x^{*} \delta_0 \int_0^t \mu'(s) ds  \notag \\
&=& x^{*} \delta_0[\mu(t) - \mu(0)] \notag \\
&>& 0. 
\end{eqnarray}
for all $t>0$ and for all $\omega\in \Omega$.

In this example, it is easier to work with the process $P$, but we can also
derive the corresponding equation for $Q$:
\begin{equation*}
Q(p,t)=-\sigma^{-1}\left(\frac{p-\mu(t)}{h\big(Z(t)\big)}\right),
\end{equation*}
where $\sigma^{-1}$ is the inverse of $\sigma$, that is, $\sigma^{-1}\big(%
\sigma(x)\big)=x$. This follows by inverting $P$, as a function of $x$,
using \eqref{eq:P} and the relation $Q\big(P(x,t),t\big)+x=0$.

\section{Conditions for Absence of Arbitrage}

\label{sec:NA}

\textbf{Assumption Q2.} \emph{The process $Q(p,\cdot)$ has representation
\begin{align}
&dQ(p,t) =\mu _{Q}(p,t)dt +\sigma_{Q}(p,t)\int_{0}^{1}b_{Q}(p,s,t)B(ds,dt),  \label{GenMod1} \\
&\int_{0}^{1}b_{Q}^{2}(p,s,t)ds =1.\   \label{GenMod3}
\end{align}
where $p\in [0,S], t\in[ 0 ,T]$, and the processes $\mu _{Q}(p,\cdot)$, $\sigma _{Q}(p,\cdot)$ and $ b_{Q}(p,s,\cdot)$ are $\mathbb{R}$-valued and $\mathcal{F}_{t}$-adapted.}

\begin{remark}
\textrm{Equality \eqref{GenMod1} is the usual semimartingale condition on
the process $Q$. While not every process \eqref{GenMod1} is monotone in $p$,
a straightforward way to ensure monotonicity is to set
\begin{equation*}
Q(p,t)=\Psi\big(p,t,B\big)
\end{equation*}
for some smooth function $\Psi$ that is strictly decreasing in the first
argument. Other examples are below in this section. }

\textrm{Condition \eqref{GenMod3} is a standard normalization, which, with
the presence of $\sigma_{Q}$, leads to no loss of generality.}
\end{remark}

\textbf{Assumption Q3.} \emph{The function $Q=Q(p,t)$ is twice continuously
differentiable with respect to $p$ for every $t$ and is continuous in $t$
for every $p$, and the function
\begin{equation*}
\tilde{\sigma}_Q(p,s,t)=\sigma_Q(p,t) b_Q(p,s,t).
\end{equation*}
is continuously differentiable with respect to $p$ for every $s$ and $t$.}

For notational convenience, we introduce the function
\begin{equation*}
C(p,t)=-\frac{1}{\frac{\partial Q}{\partial p}(p,t)} \int_{0}^{1} \frac{%
\partial \tilde{\sigma} _{Q}(p,s,t)}{\partial p} \tilde{\sigma}_Q(p,s,t) ds.
\end{equation*}

\begin{theorem}
\label{lemma::para} Suppose that Assumptions \textbf{Q1}, \textbf{Q2}, and \textbf{Q3} hold. Then, for
every $x\in [x_{\mathrm{min}},x_{\mathrm{\max}}]$, the price process $%
t\mapsto P(x,\cdot)$ satisfies
\begin{equation}  \label{P_1}
dP(x,t) =\mu _{P}(x,t)dt+\sigma _{P}(x,t)\int_{0}^{1}b_{P}(x,s,t)B(ds,dt),
\end{equation}
where
\begin{align}
&\mu _{P}(x,t) =-\frac{\mu _{Q}\big(P(x,t),t\big) +\frac{1}{2}\frac{\partial
^{2}Q}{\partial p^{2}}\big(P(x,t),t\big)\sigma _{P}^{2}\big(x,t\big) +C%
\big(P(x,t),t\big)}{\frac{\partial Q}{\partial p}\big(P(x,t),t\big)},
\label{mu_Px} \\
&\sigma _{P}(x,t) =\frac{\sigma _{Q}\big(P(x,t),t\big)}{\frac{\partial Q}{%
\partial p}\big(P(x,t),t\big)},  \label{sigma_P} \\
&b_{P}(x,s,t) =-b_{Q}\big(P(x,t),s,t\big).  \label{b_P}
\end{align}
\end{theorem}

\begin{proof}
Monotonicity of $Q$ implies that the process $P$ defined by %
\eqref{market_clearing_equation} exists and is unique, and $\frac{\partial Q%
}{\partial p}<0$, so that (i) the expression (\ref{mu_Px}) is well-defined
and (ii) $P(x,t)\in [0,S]$. Boundedness of $Q$ implies that \eqref{Int-IW}
holds.

Using (\ref{IWK}), we will now verify that \eqref{P_1} defines the required
process:
\begin{equation}  \label{ItoWE}
\begin{split}
0=dQ&\big(P(x,t),t\big)=\mu _{Q}\big(P(x,t),t\big)dt +\int_{0}^{1}\tilde{%
\sigma}_{Q}\big(P(x,t),s,t\big)B(ds,dt) \\
&+\frac{\partial Q}{\partial p}\big(P(x,t),t\big)\left( \mu _{P}(x,t)dt
+\sigma_{P}(x,t)\int_{0}^{1}b_{P}(x,s,t)B(ds,dt)\right) \\
&+\frac{\sigma^2_{P}(x,t) }{2} \frac{\partial ^{2}Q}{\partial p^{2}}\big(%
P(x,t),t\big)dt \\
&+\sigma_P(x,t)\left(\int_0^1\frac{\partial \tilde{\sigma}_Q}{\partial p}%
\big(P(x,t),s,t\big) b_P(x,s,t)ds\right)dt.
\end{split}%
\end{equation}

Setting the martingale component in \eqref{ItoWE} equal to zero yields %
\eqref{sigma_P} and \eqref{b_P}. After that, setting the drift component in %
\eqref{ItoWE} equal to zero yields \eqref{mu_Px}.
\end{proof}

Next, we investigate the conditions for existence of an equivalent
martingale measure, that is, the measure under which the price process $P$
is a martingale. For notational convenience, we define
\begin{equation}
A(p,t) =\mu _{Q}(p,t) +\frac{1}{2}\frac{\partial ^{2}Q}{\partial p^{2}}(p,t)
\left(\frac{\sigma_{Q}(p,t)}{\frac{\partial Q}{\partial p}(p,t)}\right)^{2}
+C(p,t). \label{eqn::A_p_t}
\end{equation}

\begin{definition}
The \textit{market price of risk} is a random function $\lambda=\lambda(s,t)$
such that
\begin{equation}
\int_{0}^{1}\tilde{\sigma}_Q(p,s,t )\lambda (s,t)ds =A(p,t),\ \ p\in[0,S], \
t\in [0,T],  \label{MPR}
\end{equation}
and, for every $s\in [0,1]$ and $t\in [0,T]$, the random variable $%
\lambda(s,t)$ is $\mathcal{F}_{t}$-measurable We call equation \eqref{MPR}
\textit{market price of risk equation in demand format}.
\end{definition}

\begin{remark}
\textrm{While we consider $Q$ as the main input (or primitive) for the
underlying model, it is possible, by \eqref{market_clearing_equation}, to
take $P$ instead of $Q$ as the corresponding primitive. Then, by %
\eqref{mu_Px} and \eqref{sigma_P}, formula \eqref{MPR_price} becomes
\begin{equation}
\int_{0}^{1}\tilde{\sigma}_P(x,s,t )\lambda (s,t)ds =\mu_P(x,t) ,\ \ \ x\in[%
x_{\min}, x_{\max}],\ t\in [0,T],  \label{MPR_price}
\end{equation}
with $\tilde{\sigma}_P(x,s,t )=\sigma_P(x,t)b_P(x,s,t)$. We call %
\eqref{MPR_price} the \textit{market price of risk equation in price format}.%
}
\end{remark}

\begin{theorem}
\label{thm::no_arbitrage} In addition to Assumptions \textbf{Q1}, \textbf{Q2}%
, and \textbf{Q3}, suppose that equation \eqref{MPR} has a solution $%
\lambda=\lambda(s,t)$ satisfying \eqref{E1}. Define the measure $\mathbb{Q}$
according to \eqref{MM}. Then, for every $x\in [x_{\min}, x_{\max}]$, the
price process $t\mapsto P(x,t)$ is a martingale with respect to the measure $%
\mathbb{Q} $.
\end{theorem}

\begin{proof}
By Theorem \ref{lemma::para}, the process $P$ has representation
\begin{equation*}
dP(x,t)=-\frac{A\big(P(x,t),t\big)}{\frac{\partial Q}{\partial p}\big(%
P(x,t),t\big)}\ dt -\frac{\sigma_Q\big(P(x,t),t\big)}{\frac{\partial Q}{%
\partial p}\big(P(x,t),t\big)} \int_0^1 b_Q\big(P(x,t),s,t\big)B(ds,dt).
\end{equation*}
By Corollary \ref{cor-Girsanov-diffusion}, this process is a martingale with
respect to the measure $\mathbb{Q}$ from \eqref{MM} if equality \eqref{MPR}
holds; condition \eqref{E1} ensures that the measure $\mathbb{Q}$ is
well-defined.
\end{proof}

Combining Theorem \ref{thm::no_arbitrage} with \cite[Theorem 3.3]{BB04}, we
conclude that, under Assumptions \textbf{Q1}, \textbf{Q2}, and \textbf{Q3},
there is no arbitrage in our model.

Equation \eqref{MPR} is a Fredholm integral equation of the first kind and
has a solution if and only if, for every $t\in [0,T]$, and every $\omega$, the right-hand side $%
A(p,t)$ is in the range of the integral operator
\begin{equation*}
\mathcal{A}_t: f(s)\mapsto \int_0^1 \tilde{\sigma}_Q(p,s,t)f(s)ds;
\end{equation*}
We now investigate models under which the market price of risk equations
admit a solution. In subsections 5.1 and 5.2 the primitive of the model is
the net demand $Q$, so that we investigate whether \eqref{MPR} has a
solution, while in subsection 5.3, the primitive of the model is the price
process $P$, so we investigate whether \eqref{MPR_price} has a solution.
While working with $P$ results in market price of risk equations that are
simpler to solve, we refer the reader to Remark \ref{remmodelPisbad} where
we describe a subtle drawback of working with $P$ instead of with $Q$.

\subsection{A Linear Model}

A sufficient condition for the existence of an equivalent martingale measure
is existence of a \emph{bounded} solution $\lambda$ of \eqref{MPR}. This
condition is easy to verify when the demand curve is linear in $p$; cf. \cite%
{RA11}.

\begin{proposition}
\label{prop:MPR-linear} Assume that the demand curve depends linearly on $p$
so that
\begin{equation*}
\frac{\frac{\partial \tilde{\sigma}_Q}{\partial p}(p,s,t)}{\frac{\partial Q}{%
\partial p}(p,t)}=h_Q(s,t) \quad \mathrm{\ and } \quad \frac{\partial^2 Q}{%
\partial p^2}(p,t)=0
\end{equation*}
with a bounded function $h_Q$. If  there exists a bounded function $%
\lambda_Q=\lambda_Q(s,t)$ such that, for every $t\in [0,T]$,
\begin{equation}  \label{MPR-c1}
\mu_Q(p,t)=\int_0^1 \tilde{\sigma}_Q(p,s,t) \lambda_Q(s,t) ds,
\end{equation}
then equation \eqref{MPR} has a bounded solution
\begin{equation}  \label{MPR-lin-sol}
\lambda(s,t)=\lambda_Q(s,t)-h_Q(s,t).
\end{equation}
\end{proposition}

\begin{proof}
Under the assumptions of the proposition, equation \eqref{MPR} becomes
\begin{equation*}
\int_0^1 \tilde{\sigma}_Q(p,s,t)\lambda(s,t)ds= \int_0^1 \tilde{\sigma}%
_Q(p,s,t)\lambda_Q(s,t)ds -\int_0^1 \tilde{\sigma}_Q(p,s,t) h_Q(s,t)ds,
\end{equation*}
and then \eqref{MPR-lin-sol} follows.
\end{proof}

Combining Proposition \ref{prop:MPR-linear} with \cite[Theorem 2.6]{RA11},
we conclude that, in the case of the linear demand curve, there are no
arbitrage opportunities as long as condition \eqref{MPR-c1} holds. Note that %
\eqref{MPR-c1} is  similar to Condition C.4 in the original HJM model \cite%
{HJM-orig}.

\subsection{A Separated Model}

An extension of a linear demand model is a \emph{separated} demand curve
\begin{equation*}
Q(p,t)=\sigma(p)F\left(\int_0^t\int_0^1b(s,u)B(ds,du)\right).
\end{equation*}

\begin{proposition}
\label{prop:sep} Assume that
\begin{equation*}
\int_0^1b^2(s,u)ds=1,
\end{equation*}
for all $u\in [0,S]$, and also
\begin{align}  \label{sigmaODE}
& 0<\delta \leq F(x), 0<\delta \leq |F^{\prime }(x)|\leq C,\ |F^{\prime \prime }(x)|\leq C, \notag \\
& \sigma^{\prime }(p)<0, |b(t,u)| \leq C, \notag \\
& \frac{d}{dp}\left(\frac{\sigma^{\prime \prime }(p)\sigma(p)}{\big(%
\sigma^{\prime }(p)\big)^2}\right)=0.
\end{align}
Then equation \eqref{MPR} has a bounded solution.
\end{proposition}

\begin{proof}
By the It\^{o} formula,
\begin{align*}
&\mu_Q(p,t)=\sigma(p)h_2(t),\ \ b_Q(t,s)=b(t,s), \\
&\sigma_Q(p,t)=\sigma(p) h_1(t),\ \ \tilde{\sigma}_Q(p,t,s)=%
\sigma(p)b(s,t)h_1(t),
\end{align*}
with
\begin{eqnarray*}
h_1(t) &=& F^{\prime }\left(\int_0^t\int_0^1b(s,u)B(ds,du)\right), \\
h_2(t) &=& \frac{1}{2}F^{\prime \prime }\left(\int_0^t\int_0^1b(s,u)B(ds,du)\right).
\end{eqnarray*}
By direct computation, the bounded solution of equation \eqref{MPR} is
\begin{equation*}
\lambda(s,t)=b(s,t)\left(\frac{h_2(t)}{h_1(t)} +\frac{h_1(t)\sigma_0}{2h_0(t)} - \frac{h_1(t)}{%
h_0(t)}\right),
\end{equation*}
with
\begin{equation*}
h_0(t)=F\left(\int_0^t\int_0^1b(s,u)B(ds,du)\right),
\end{equation*}
and
\begin{equation*}
\sigma_0=\frac{\sigma^{\prime \prime }(p)\sigma(p)}{\big(\sigma^{\prime }(p)%
\big)^2}.
\end{equation*}
\end{proof}

Equation \eqref{sigmaODE} defines a three-parameter family of functions $%
\sigma$. This family includes the linear function, corresponding to $%
\sigma_0=0$.

\subsection{A Lognormal Model}

\label{sec::lognormal_model}

In this model, the input is $P$, the price as a function of quantity, and
the main objective  becomes analysis of the market price of risk equations
in price format \eqref{MPR_price}. To solve the resulting Fredholm equation
of the first kind, we will transform it to a Volterra equation.

Let $0<\varepsilon<1/2$. We define a scale function mapping the quantity
variable $x$ to the noise variable $s$:
\begin{equation}
f(x)=2\varepsilon+(1-2\varepsilon)\frac{x-x_{\min}}{x_{\max}-x_{\min}};
\end{equation}
note that $2\varepsilon \leq f(x)\leq 1$ for $x\in [x_{\min},x_{\max}]$.

Introduce non-random functions $p_0=p_0(x)$, $\bar{\mu}_p=\bar{\mu}_p(x)$
and $\bar{\sigma}=\bar{\sigma}_{p}(x,s)$ and suppose that $p_0$ and $\bar{\mu%
}_{p}$ are $\mathcal{C}^1$ (that is, bounded and continuously
differentiable, with a bounded derivative) and $\bar{\sigma}$ is $\mathcal{C}%
^1$ in $x$, uniformly with respect to $s$. We also assume that

\begin{itemize}
\item $\bar{\sigma}_{p}(x,s)=0$ for $s\leq \varepsilon;$

\item $\bar{\sigma}_{p}\big(x,f(x)\big)$  is uniformly bounded from zero for
$x\in [x_{\min},x_{\max}]; $

\item $\bar{\sigma}_{p}(x,s)=0$ for $s>f(x)$;\newline
${\ } $

\item $\int_{0}^{\varepsilon }\bar{\sigma}_{p}(x_{\min },s)ds\neq 0$;
\newline
${\ } $

\item $\int_{\varepsilon }^{2\varepsilon }\bar{\sigma}_{p}(x_{\min},s)ds\neq
0$.
\end{itemize}

Define the \emph{price density function}
\begin{equation}  \label{prd}
p(x,t)=p_0(x)\exp \left(\bar{\mu}_{p}(x)t+\int_{\varepsilon }^{1}\bar{\sigma}%
_{p}(x,s)B(ds,t) -{\frac{1}{2}}\int_{\varepsilon}^1 \bar{\sigma}_p^2(x,s)ds \right),
\end{equation}

and then the price process

\begin{eqnarray}  \label{PPr}
P(x,t) &=&p(x_{\min },t)+\int_{x_{\min }}^{x}p(y,t)dy, \quad x\in
[x_{\min},x_{\max}].
\end{eqnarray}

\begin{proposition}
\label{prop:expo} Equation \eqref{MPR_price} has a bounded solution.
\end{proposition}

We prove Proposition \ref{prop:expo} in Section \ref{sec::A1}.

\begin{remark}
\label{rm:SPDE} \textrm{{It is possible to represent the net demand curve
\textquotedblleft in the risk-neutral measure\textquotedblright , like
practitioners do to model interest rates in the HJM framework: combining %
\eqref{mu_Px} and \eqref{sigma_P} with Theorem \ref{thm::no_arbitrage} shows that the process $%
Q$ under the measure $\mathbb{Q}$ satisfies
\begin{eqnarray}
dQ(p,t) &=&\left(-\frac{1}{2}\frac{\partial ^{2}Q(p,t)}{\partial p^{2}}\left(
\frac{\sigma _{Q}(p,t)}{\frac{\partial Q(p,t)}{\partial p}}\right) ^{2}+%
\frac{1}{2\frac{\partial Q}{\partial p}(p,t)}\frac{\partial (\sigma _{Q}(p,t))^{2}%
}{\partial p} \right)dt \label{Q-SPDE} \notag \\
& & \qquad + \sigma _{Q}(p,t)\int_{0}^{1}b_{Q}(p,s,t)B^{\mathbb{Q}}(ds,dt)\text{,}
  \label{SPDE} \\
Q(p,0) &=&Q_{0}(p).  \notag
\end{eqnarray}%
In this work we always define $Q$ with respect to the physical measure $%
\mathbb{P}$, which, under the conditions of Theorem \ref{thm::no_arbitrage},
automatically leads to a classical solution of \eqref{Q-SPDE}. On the other
hand, being a (rather complicated) quasi-linear stochastic partial
differential equation, \eqref{Q-SPDE} is ill-posed in the sense of Hadamard;
cf. \cite[Section 3.7]{Chow-SPDE}. This ill-posedness suggests that
liquidity models are, in general, unstable, demonstrating  yet another
benefit of formulating the model in the form $Q=Q(p,t)$, that is, the
quantity as a function of price. Indeed, instability of the model is not at
all obvious under the formulation  $P=P(x,t)$, that is, price as a function
of quantity, used in most of existing literature  on the subject. We suspect
that the reason for instability is that monotonicity, of either $P$ or $Q$,
is a hard condition to meet; without monotonicity, it is harder to connect
the processes $P$ and $Q$ and to make the family of prices a $\mathbb{Q}$%
-martingale. } }
\end{remark}

\section{Pricing of Options in a Manipulable Market}

\label{sec:PMC} As before, our market consists of atomistic traders and one
large trader. The question we address in this section is how to characterize
the price of a derivative security when the large trader can manipulate the
market.  While it might appear that the price of the derivative security
should depend on the future of the large trader's strategy $\theta$, this
section shows that it is  often not the case.

Denote by $\Theta^{BV}$ the subset of $\Theta$ (the set of the large trader
strategies) consisting of all the functions with bounded variation
\begin{equation*}
\theta(t)=\theta(0)+\int_0^t \dot{\theta}(s)ds,\ \ \dot{\theta}\in
L_1((0,T)).
\end{equation*}
By Corollary \ref{cor:SPM} (cf. \cite[page 7]{BB04}), absence of transaction
costs for the large trader is equivalent to the condition $\theta\in
\Theta^{BV}$.

We define the observable net demand on the market by
\begin{equation}
{Q}^{\theta}(p,t)=Q(p,t)+\theta (t),  \label{Q_theta_def}
\end{equation}
so that
\begin{align*}
&d{Q}^{\theta}(p,t) =\big(\mu_{{Q}}(p,t)+\dot{\theta}(t)\big)dt +\int_{0}^{1}%
\tilde{\sigma}_{{Q}}( p,s,t)B(ds,dt),\ \ \ p\in [0,S],\ t\in[ 0 ,T].
\end{align*}

Likewise, the corresponding price process $P^{\theta}={P}^{\theta}(x,t)$ is
defined by
\begin{equation}  \label{P_theta_def}
{Q}^{\theta}\big({P}^{\theta}(x,t),t\big)+x=0,
\end{equation}
and then, by the It\^o-Wentzell formula,
\begin{align}  \label{P-theta}
d{P}^{\theta}(x,t) &= \left(\mu_{{P}}(x,t)+\frac{\partial P(x,t)}{\partial x}%
\, \dot{\theta}(t)\right)dt \notag \\
& \quad -\int_{0}^{1}\frac{\tilde{\sigma}_{Q}\big({P}^{\theta}(x,t),s,t\big)}
{\frac{\partial Q}{\partial p} \big({P}^{\theta}(x,t),t\big)}B(ds,dt),\ t\in[ 0 ,T].
\end{align}
The quantity $x$ represents the deviation of the trader's position from the
strategy $\theta$; the range of admissible values of $x$ will, in general,
depend on $\theta$.

We denote the clearing price at time $t$ by ${\pi}^{\theta}(t)$:
\begin{equation*}
{Q}^{\theta}\big({\pi}^{\theta}(t),t\big)=0.
\end{equation*}

The constant strategy%
\begin{equation*}
\theta^{c}(t):=\theta(0), \qquad 0\leq t \leq T;
\end{equation*}
will be of special interest.

\textbf{{Assumption Q4}.} For every $\theta \in \Theta^{BV}$, there exists a
measure $\mathbb{Q}^{\theta}$  such that, for every admissible $x$,  the
process $t\mapsto {P}^{\theta}(x,t)$ is a martingale  under $\mathbb{Q}%
^{\theta}$.

\begin{theorem}
\label{thm::Q_price} If $\theta \in \Theta^{BV}$, then, for every Borel set $%
A\in \mathcal{B}(\mathcal{C}([0,T]))$,%
\begin{equation*}
\mathbb{Q^{\theta }}\big({\pi}^{\theta}\in A\big) =\mathbb{Q}^{\theta^c }%
\big(\pi^{\theta^c }\in A\big).
\end{equation*}
\end{theorem}

\begin{proof}
Consider the stochastic process $X=X(t)$ defined by the equation
\begin{equation*}
dX(t)=-\int_{0}^{1}\frac{\tilde{\sigma}_{Q}(X(t),s,t)}
{\frac{\partial Q}{\partial p} (X(t),t)}B(ds,dt),
\end{equation*}
with initial condition $X(0)=P^{\theta}(0,0)$; the properties of $%
\tilde{\sigma}_Q$ imply existence and uniqueness of the solution.

By construction,
\begin{equation}  \label{pi-theta}
\pi^{\theta}(t)=P^{\theta}(0,t).
\end{equation}
Switching from the original measure $\mathbb{P}$ to the measure $\mathbb{%
Q^{\theta }}$ removes the drift part in \eqref{P-theta} but does not change
the diffusion part. As a result, \eqref{P-theta} and \eqref{pi-theta} imply
\begin{equation*}
\mathbb{Q^{\theta }}\big({\pi}^{\theta}\in A\big)= \mathbb{P}(X\in A)=%
\mathbb{Q}^{\theta^c }\big(\pi^{\theta^c }\in A\big),\ A\in \mathcal{B%
}\big(\mathcal{C}([0,T])\big),
\end{equation*}
completing the proof.
\end{proof}

We seek to price a contingent claim $H^{\theta }$of the form:
\begin{equation*}
H^{\theta }=F\big(\pi ^{\theta }(\cdot )\big),
\end{equation*}%
for some continuous functional $F$ on $\mathcal{C}((0,T))$. We assume that
it can be replicated by a trading strategy $\Delta ^{\theta }\in \Theta $.
Theorem 4.1 in \cite{BB04} shows that the asymptotic liquidation process
generated by a strategy $\Delta ^{\theta }\in \Theta $ can be $\varepsilon -$%
approximated by a strategy $\Delta ^{\theta ,\varepsilon }\in \Theta ^{BV}$.
By proposition 3.3, the corresponding realizable wealth $V^{\Delta ^{\theta
,\varepsilon }}$ is a $\mathbb{Q}^{\theta }$-martingale, and thus we can
define (up to an $\varepsilon -$approximation) the no-arbitrage price of the
contingent claim $H^{\theta }$ by:%
\begin{equation*}
\nu _{H}^{\theta }=\mathbb{E}^{\mathbb{Q}^{\theta }}[F(\pi ^{\theta })].
\end{equation*}%
\bigskip

The following theorem is the main result of this section.

\begin{theorem}
\label{th:PCC} Suppose that Assumptions \textbf{Q1-Q4} hold. If $\theta \in
\Theta ^{BV}$, then the no-arbitrage price of $H^{\theta }$ depends only on
the initial value of the large trader strategy:
\begin{equation}
\nu _{H}^{\theta }=\mathbb{E}^{\mathbb{Q}^{\theta ^{c}}}[F(\pi ^{\theta
^{c}})].  \label{OPF}
\end{equation}
\end{theorem}

\begin{proof}
With $\theta \in \Theta ^{BV}$, \eqref{OPF} now immediately follows from
Theorem \ref{thm::Q_price}.
\end{proof}

\begin{remark}
\textrm{{\label{remmodelPisbad} Theorem \ref{th:PCC} shows that working with
$P$ as opposed to $Q$ has a subtle drawback. Suppose without loss of
generality that today's position of the large trader is $\theta(0)=x_{\min}$%
. In the lognormal model of Section \ref{sec::lognormal_model}, the clearing
price $\pi$ will thus be equal to $p(x_{\min})$. Looking at the definition
of $p(x_{\min})$ in Equation (\ref{prd}), we see that the remainder of the
price curve, i.e., $p(x)$ for $x > x_{\min}$ has no influence on the
dynamics of $p(x_{\min})$, and we are back to the standard Black-Scholes
model. Liquidity effects would be present if we specified volatility to be
stochastic, but this would further complicate the market price of risk
equations.} }
\end{remark}

\section{A Practical Model}

\label{sec:MOD}

\subsection{A Continuous Version}

The model in this section satisfies the main requirements for the use of the
option pricing formula developed in the previous section. In particular, the
net demand curve is specified as quantity as a function of price; this
allows for a better modeling of liquidity effects than models of price as a
function of quantity will allow (see Remark \ref{remmodelPisbad}).

We specify our model as follows:%
\begin{eqnarray}
Q(0,t) &=&Q(0,0)+F_{0}\left(\int_{0}^{1}\bar{\sigma}_{Q}(0,s)B(ds,t)
\right),  \label{Mod1} \\
q(p,t) &=&q(p,0)+F_{1} \left(\int_{0}^{1}\bar{\sigma}_{q}(p,s)B(ds,t)
\right),  \label{Mod2} \\
Q(p,t) &=&Q(0,t)-\int_{0}^{p}q(x,t)dx,  \label{Mod3}
\end{eqnarray}
and assume that

\begin{itemize}
\item $0 < \delta_{0,\min} \leq F_0(x) \leq \delta_{0,\max}$, $0<\delta_{1,\min} \leq F_1(x)
\leq \delta_{1,\max}$;

\item $F_0\in \mathcal{C}^1,\ \inf_x|F^{\prime }_0(x)|>0$;

\item $F_1\in \mathcal{C}^3,\ F_{1}^{\prime }(x)\leq -\varepsilon_1<0$;

\item $|\bar{\sigma}_{Q}(0,s,t)|+|\bar{\sigma}_{q}(p,s,t)|+|\partial \bar{%
\sigma}_{q}(p,s,t)/\partial p|\leq C_1$;

\item There exists an $\varepsilon \in (0,1)$ such that $\bar{\sigma}%
_{q}(p,s,t)=0$ for $(s,p)\in \lbrack 0,\varepsilon ]\times \lbrack
\varepsilon S,S]$, and:
\begin{equation}
Q(\varepsilon S,t)+x_{\min }\geq 0,\qquad 0\leq t\leq T.  \label{Bound_hold}
\end{equation}
\end{itemize}

Also, $Q(0,t)$ is positive by construction, and so condition
\begin{equation}
Q(S,t)+x_{\max }\leq 0,\qquad 0\leq t\leq T,  \label{MC_hold}
\end{equation}%
must hold for the clearing price to be less than $S$. Let:%
\begin{eqnarray*}
h_{0}(t) &\equiv& \frac{\partial F_0}{\partial x} \left( \int_{0}^{1}\bar{\sigma}%
_{Q}(0,s,t)B(ds,t)\right),  \\
h_{1}(p,t) &\equiv& \frac{\partial F_1}{\partial x} \left( \int_{0}^{1}\bar{\sigma}%
_{q}(p,s,t)B(ds,t)\right).
\end{eqnarray*}

One convenient way (which we will use in the next section to get explicit formula) for the market price of risk equations to hold, is to assume that:%
\begin{equation}
h_{0}(t)\int_{0}^{\varepsilon }\bar{\sigma}_{Q}(0,s,t)ds-h_{1}(\varepsilon
S,t)\int_{0}^{\varepsilon }\int_{0}^{\varepsilon S}\bar{\sigma}%
_{q}(x,s,t)dxds\geq \eta >0.  \label{positivity}
\end{equation}

\begin{lemma}
\label{lem::model_hold}  Conditions \eqref{Bound_hold}, \eqref{MC_hold}, and
$0<\varepsilon <S$ can be satisfied with a suitable choice of the
parameters $\delta _{0,\min}, \delta _{0,\max}, \delta _{1,\min}, \delta _{1,\max}, x_{\min},
x_{\max}, S$ in \eqref{Mod1}, \eqref{Mod2}, \eqref{Mod3}.
\end{lemma}

We prove Lemma \ref{lem::model_hold} in Section \ref{sec::A2};  Figure \ref%
{fig::liq_illus} provides an illustration.

\begin{center}
\begin{figure}[th]
\centering
\includegraphics[scale = 0.6]{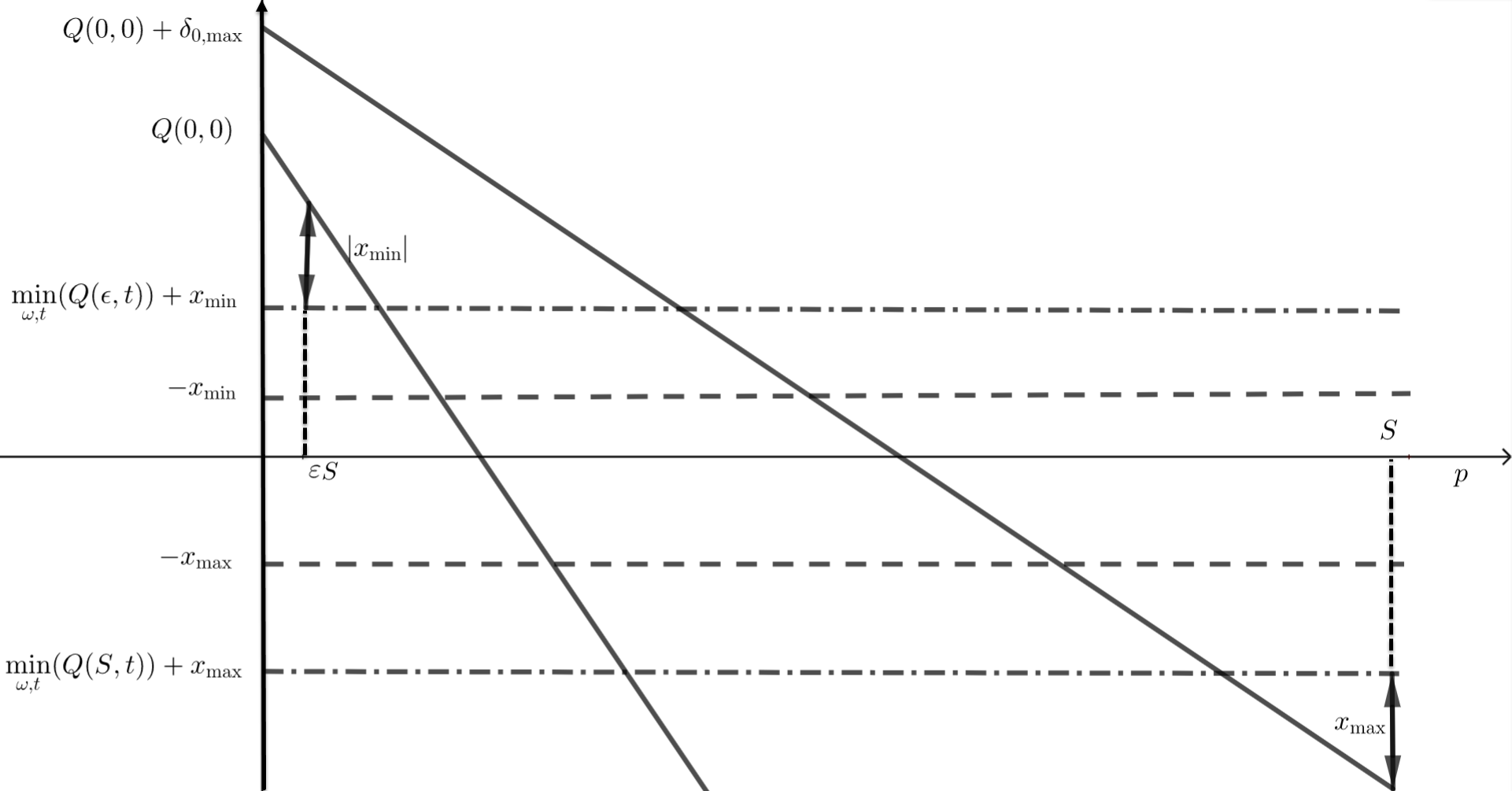} \newline
\caption{Illustration of liquidity model variables. }
\label{fig::liq_illus}
\end{figure}
\end{center}

By (7.4), the clearing price is bounded below by $\varepsilon S$. Thus,
there will no be arbitrage if the market price of risk equations (\ref%
{MPR_disc}) hold for the smaller domain $p\in \lbrack \varepsilon S,S]$:

\begin{equation}
\int_{0}^{1}\left( \sigma _{Q}(0,s,t)-\int_{0}^{p}h_{1}(x,t)\bar{%
\sigma}_{q}(x,s,t)dx\right) \lambda (s,t)ds = A(p,t).  \label{MPR_disc}
\end{equation}

Differentiating (\ref{MPR_disc}) with respect to $p$ we obtain, since $\bar{%
\sigma}_{q}(p,s,t)=0$ for $(s,p)\in \lbrack 0,\varepsilon ]\times \lbrack
\varepsilon S,S]$:%
\begin{equation}
\int_{\varepsilon }^{1}h_{1}(p,t)\bar{\sigma}_{q}(p,s,t)\lambda (s,t)ds=-%
\frac{\partial A(p,t)}{\partial p},\qquad p\in \lbrack \varepsilon S,S].
\label{MPR_part}
\end{equation}

Since, by assumption, $h_{1}(p,t)\leq \eta <0$ we can now divide by $%
h_{1}(p,t)$ to obtain the following deterministic Fredholm equation of the
first kind:%
\begin{equation}
\int_{\varepsilon }^{1}\bar{\sigma}_{q}(p,s,t)\lambda (s,t)ds=-\frac{%
\partial A(p,t)}{\partial p}\frac{1}{h_{1}(p,t)},\qquad p\in \lbrack
\varepsilon S,S].  \label{Fred}
\end{equation}

\begin{proposition}
Suppose that the Fredholm equation (\ref{Fred}) has a unique solution for $%
s\in \lbrack \varepsilon ,1]$.Then there exists a unique solution to the
market price of risk equations.
\end{proposition}

\begin{proof}
We need to prove that there is a solution to the equation (\ref{MPR_disc})\
for $p=\varepsilon S$. This equation can be rewritten as:%
\begin{equation*}
\int_{0}^{\varepsilon }\tilde{\sigma}_{Q}(\varepsilon S,s,t)\lambda
(s,t)ds=A(\varepsilon S,t)-\int_{\varepsilon }^{1}\tilde{\sigma}%
_{Q}(\varepsilon S,s,t)\lambda (s,t)ds.
\end{equation*}

We observe that:%
\begin{equation*}
\int_{0}^{\varepsilon }\tilde{\sigma}_{Q}(\varepsilon
S,s,t)ds=\int_{0}^{\varepsilon } \big( h_{0}(t)\bar{\sigma}_{Q}(0,s,t)-h_{1}(%
\varepsilon S,t)\int_{0}^{\varepsilon S}\bar{\sigma}_{q}(x,s,t) \big)dxds.
\end{equation*}

Thus, by (\ref{positivity}) $\int_{0}^{\varepsilon }\tilde{\sigma}%
_{Q}(\varepsilon S,s,t)ds\geq \eta >0$. Thus, for $s\in \lbrack
0,\varepsilon )$ we can choose the constant solution:%
\begin{equation}
\lambda (s,t)=\frac{A(\varepsilon S,t)-\int_{\varepsilon }^{1}\tilde{\sigma}%
_{Q}(\varepsilon S,s,t)\lambda (s,t)ds}{\int_{0}^{\varepsilon }\tilde{\sigma}%
_{Q}(\varepsilon S,s,t)ds}.  \label{tagada}
\end{equation}
\end{proof}

However, it is not possible to establish practical conditions on the kernel $%
\bar{\sigma}_{q}$ such that (\ref{Fred})\ has a solution for any right-hand
side. This problem does not occur in the discrete case, so we move to a
finite dimensional approximation of our problem.

Recall that the formula of $A(p,t)$ is provided in Equation \ref{eqn::A_p_t}. It is sufficient to obtain $A(p,t)$ with $\mu_Q(p,t)$, $\sigma_Q(p,t)$. We provide the discretized version of the model in the next subsection.

\subsection{A Discretized Version of the Model} \label{sec::dis_mod}

We consider a finite trading interval $[0,T]$ and discretize our time axis $%
0=t_{0}^{(n)}<\ldots <t_{J_{n}}^{(n)}=T$. \ We will verify the market price
of risk equations only for the set of discrete prices $\mathcal{P}^{(n)}$
where
\begin{equation*}
\mathcal{P}^{(n)}=\Big\{p_{1}^{(n)}\leq p_{2}^{(n)}\leq \cdots \leq
p_{I_{n}}^{(n)}\Big\},
\end{equation*}

with $p_{1}^{(n)}=\varepsilon S$ and $p_{I_{n}}^{(n)}=S$. We also set $%
p_{0}^{(n)}=0$. Finally we discretize the space of factors by setting, for $%
k=0,..,I_{n}:$

\begin{equation*}
s_{k}^{(n)}=\frac{k}{I_{n}}.
\end{equation*}

In our discrete model, the function $\bar{\sigma}_{Q}(0,.,.)$ is
approximated by:
\begin{eqnarray*}
\bar{\sigma}_{Q}^{(n)}(0,s,t) = \sum_{j=0}^{J_{n}-1}\mathbf{1}_{[t_{j}^{(n)},t_{j+1}^{(n)})} (t) \times \sum_{k=0}^{I_{n}-1} \mathbf{1}_{[s_{k}^{(n)},s_{k+1}^{(n)})}(s) \bar{\sigma}_{Q}\left(0,s_{k}^{(n)},t_{j}^{(n)}\right).
\end{eqnarray*}

To simplify notation, we define the approximation of the function $\bar{%
\sigma}_{q}$ in a slightly different way, whether $(s,p)\in \lbrack
0,\varepsilon )\times \lbrack \varepsilon S,S]$ or whether $(s,p)$ belongs
to the remainder of the domain. Indeed, for the former domain, we need to
introduce a collection of invertible matrices $\bar{\sigma}_{q,j}^{(n)}$ for
$j\in \lbrack 0,J_{n}-1]$:%
\begin{equation*}
\bar{\sigma}_{q,j}^{(n)}=\left[
\begin{array}{ccc}
\bar{\sigma}_{q,j}^{(n)}(1,1) & ... & \bar{\sigma}_{q,j}^{(n)}(1,I_{n}-1) \\
\vdots &  & \vdots \\
\bar{\sigma}_{q,j}^{(n)}(I_{n}-1,1) & ... & \bar{\sigma}%
_{q,j}^{(n)}(I_{n}-1,I_{n}-1)%
\end{array}%
\right].
\end{equation*}

The function $\bar{\sigma}_{q}$ is then approximated by:%
\begin{eqnarray*}
&& \bar{\sigma}_{q}^{(n)}(p,s,t) =\sum_{j=0}^{J_{n}-1}\mathbf{1}%
[t_{j}^{(n)},t_{j+1}^{(n)})(t) \times \sum_{i=0}^{I_{n}-1}\mathbf{1}%
[p_{i}^{(n)},p_{i+1}^{(n)})(p) \times  \\
&& \quad \left[\mathbf{1}[s_{0}^{(n)},s_{1}^{(n)})(s)\bar{\sigma}_{q}(0,p_{i}^{(n)},t_{j}^{(n)})+ \sum_{k=1}^{I_{n}-1}\mathbf{1}[s_{k}^{(n)},s_{k+1}^{(n)})(s)\bar{\sigma}%
_{q,j}^{(n)}(i,k)\mathbf{1}[p_{1}^{(n)},p_{I_{n}}^{(n)})(p)\right].
\end{eqnarray*}

The attentive reader will observe that, without loss of generality, $\bar{%
\sigma}_{q}^{(n)}(p,s,t)=0$ for $(s,p)\in $ $[\varepsilon ,1)\times \lbrack
0,\varepsilon S)$. We discretize the market price of risk by:%
\begin{equation*}
\lambda^{(n)} (s,t)=\sum_{j=0}^{J_{n}-1}\mathbf{1}_{[t_{j}^{(n)},t_{j+1}^{(n)})}(t) \sum_{k=0}^{I_{n}-1}\mathbf{1}_{[s_{k}^{(n)},s_{k+1}^{(n)})}(s)\lambda
_{j}^{(n)}(k).
\end{equation*}

The deterministic Fredholm equation of the first kind (\ref{Fred})\ becomes
then:

\begin{eqnarray}
& & \left[
\begin{array}{ccc}
\bar{\sigma}_{q,j}^{(n)}(1,1) & ... & \bar{\sigma}_{q,j}^{(n)}(1,I_{n}-1) \\
\vdots &  & \vdots \\
\bar{\sigma}_{q,j}^{(n)}(I_{n}-1,1) & ... & \bar{\sigma}%
_{q,j}^{(n)}(I_{n}-1,I_{n}-1)%
\end{array}%
\right] \left[
\begin{array}{c}
\lambda _{j}^{(n)}(1) \\
\vdots \\
\lambda _{j}^{(n)}(N-1)%
\end{array}%
\right] \notag \\
& & \qquad \qquad \qquad \qquad \qquad \qquad \qquad = I_{n}\left[
\begin{array}{c}
\frac{A(p_{1}^{(n)},t_{j}^{(n)})}{h_{1}(p_{1}^{(n)},t_{j}^{(n)})}-\frac{%
A(p_{2}^{(n)},t_{j}^{(n)})}{h_{1}(p_{2}^{(n)},t_{j}^{(n)})} \\
\vdots \\
\frac{A(p_{I_{n}-1}^{(n)},t_{j}^{(n)})}{h_{1}(p_{I_{n}-1}^{(n)},t_{j}^{(n)})}%
-\frac{A(p_{I_{n}}^{(n)},t_{j}^{(n)})}{h_{1}(p_{I_{n}}^{(n)},t_{j}^{(n)})}%
\end{array}\right].   \label{MPR_2}
\end{eqnarray}

which has a solution since $\bar{\sigma}_{q,j}^{(n)}$ is assumed invertible
for each $j\in \lbrack 0,J_{n-1}]$. 

\section{Empirical Analysis}

\label{sec:NumericalEx}

We calibrated the market model to historical data, and then simulated the
net demand surface, first in the physical measure, and then in the
risk-neutral measure in order to price options.

For notational simplicity, we drop the superscript $^{(n)}$ from all
variables in the last section. We choose the following model:

\begin{eqnarray*}
F_{0}(x) &=&\delta _{0,\min }+\frac{\delta _{0,\max }-\delta _{0,\min }}{1+x}%
, \\
F_{1}(x) &=&\delta _{1,\min }+\frac{\delta _{1,\max }-\delta _{1,\min }}{1+x}%
.
\end{eqnarray*}

\subsection{Data}

The trading data are collected from the NYSE Arca limit book orders for April 2011. Historical NYSE Arca book data provide information of the complete limit order book (LOB) from NYSE, NYSE Arca, NYSE MKT, NASDAQ and the ArcaEdge platforms from 3:30 a.m. to 8:00 p.m. ET under high speed of latencies (less than 5 milliseconds). In the empirical study of this paper, we consider the orders only from 9:30 a.m. to 4:00 p.m., when the price formation is effective and the equities are regularly and actively traded. Each limit order contains the unique reference number, the time stamp in seconds and milliseconds, the limit price in U.S. dollars, the quantity in number of shares, and the trading type (``B'': buy or ``S'' : sell).

All the limit order book records are categorized into three groups: ``A'': Add, ``M'': Modified and ``D'': Deleted. For the market liquidity model, we consider the net demand of the stock, which is captured by summing added records (``A'') with modified (``M'') adjustment and subtracting the deleted (``D'') orders. To be specific, within a certain partitioned time period, the added orders would be updated by the modified orders, if applicable and the orders occur in the same partitioned period.

We use the limit order book data of Apple Inc. (NYSE:AAPL) as of April 1st, 2011. Figure \ref{fig::Q_and_q} illustrate the net demand surface $Q$ (left) and net demand density $q$ (right) with respect to time and price. Each time step is equivalent to 5 minutes in real trading time.

\begin{center}
\begin{figure}[ht]
\centering
\includegraphics[scale = 0.135]{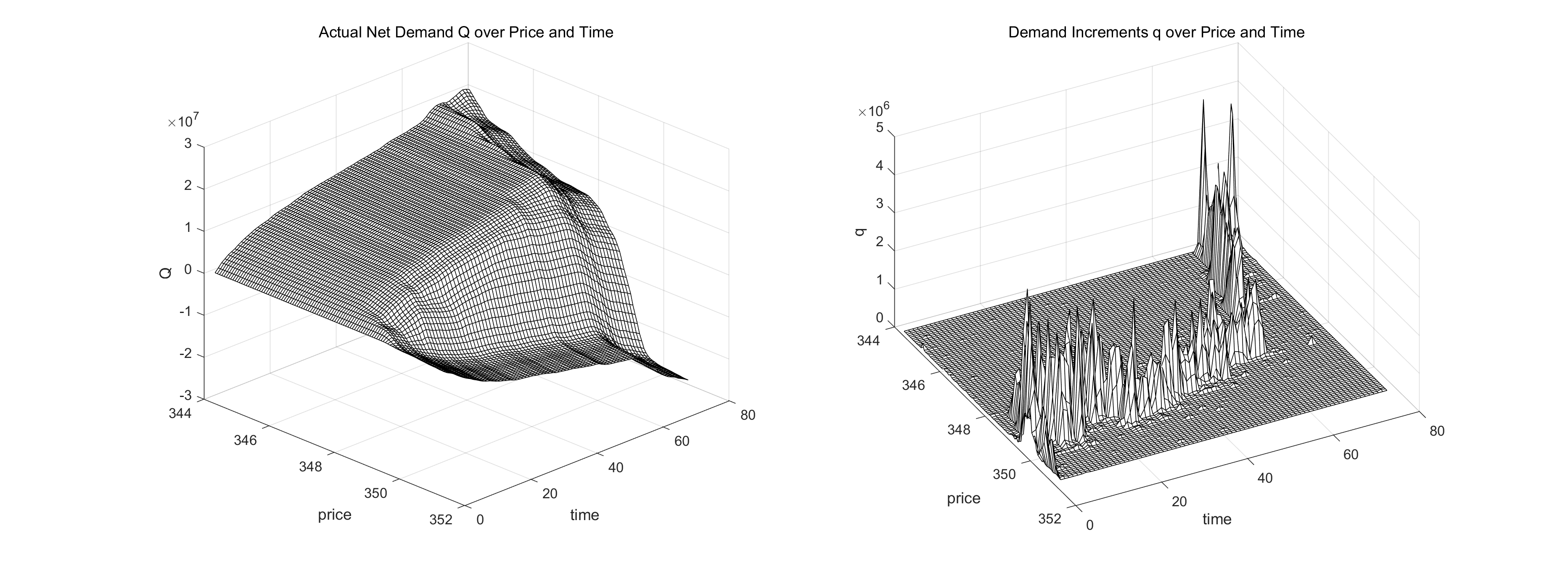} \newline
\caption{\textbf{Real dataset:} the left graph illustrates the net demand surface $Q$, and the right graph plots the corresponding net demand increments $q$ of stock of Apple Inc. (NYSE: AAPL). The market date is April 1st, 2011. Each time step represents 5 minutes in real trading time. Each price step is \$0.0672, and the price range is from \$344.24 to \$350.96. }
\label{fig::Q_and_q}
\end{figure}
\end{center}

As the theory suggests, the net demand $Q$, at any given time, decreases
monotonically with stock price. For a given price level, the net demand has
a \textquotedblleft hump\textquotedblright\ shape with time elapsing. A
plausible explanation of the hump shape is that for this particular market
date, the buy orders accumulate at the beginning of trading day, and then
large amount of sell order quickly decreases the net demand in the last few
hours.

\subsection{Calibration Methodology}

For simplicity of exposition we assume equal partition, i.e.:%
\begin{eqnarray*}
\Delta s &=&s_{k+1}-s_{k}=\frac{1}{I_{n}} \\
\Delta t &=&t_{j+1}-t_{j}=\frac{\tau }{J_{n}} \\
\Delta p &=&p_{i+1}=p_{i}=\frac{S}{I_{n}}
\end{eqnarray*}

where $\tau $ is our calibration window. In our empirical study, $S=350.96$, $I_n=100$, $J_n=78$, and $\tau=1$. We set 9:30 AM of April 1st, 2011 as time zero ($t_0=0$). First, the estimated net demand quantity is calculated by the
following formula.%
\begin{equation}
\hat{Q}(p_{i},t_{j})=\hat{Q}_{B}(p_{i},t_{j})-\hat{Q}_{S}(p_{i},t_{j}),\quad
i=0,\ldots ,I_{n}\text{ and }j=0,\ldots ,J_{n}.  \label{eqn::est_Q}
\end{equation}

where $\hat{Q}_{B}(p_{i},t_{j})$ is the quantity available for buy orders
with price greater than $p_{i}$ at time $t_{j}$, and $\hat{Q}%
_{S}(p_{i},t_{j})$ is the quantity available for sell orders with price less
than $t$ at time $t$. The net demand density $q(p_{i},t_{j})$ is estimated
as:%
\begin{equation}
\hat{q}(p_{i},t_{j})=-\frac{\hat{Q}(p_{i},t_{j})-\hat{Q}(p_{i-1},t_{j})}{%
\Delta p}.  \label{eqn::est_q}
\end{equation}

The other model parameters are $\delta_{0,\min} = 76,942$, $\delta_{0,\max} = 21,067,319$, $\delta_{1,\min} = 0$, and $\delta_{1,\max} = 1,822,500$, which satisfies all model assumptions for the stock and date we selected. The matrices $\bar{\sigma}_q$ and $\bar{\sigma}_Q$ are estimated using the method of moments.

\subsection{Simulation in the Risk Neutral Measure}

We used the Euler scheme to simulate the net demand surface. We compute the market price of risk using Equation (\ref{MPR_2}) then simulate the path with stochastic string using Equation (\ref{eqn::risk_lambda}). 

Figure \ref{fig::sim_Q_and_q} plots a simulated net demand and demand density, using one random sample. We can see from this figure that the main properties of the net demand surface are satisfied. In other words, the net demand curve $Q$ at any given time is downward sloping, and the net demand density $q$ are all positive. Furthermore, the simulated net demand surface ($Q(p,t)$) mimics the ``hump'' shape as shown in the real data. The net demand density ($q(p,t)$) concentrates on similar clearing prices and trading time as shown in the empirical limit order dataset. 

\begin{center}
\begin{figure}[h]
\centering
\includegraphics[scale = 0.135]{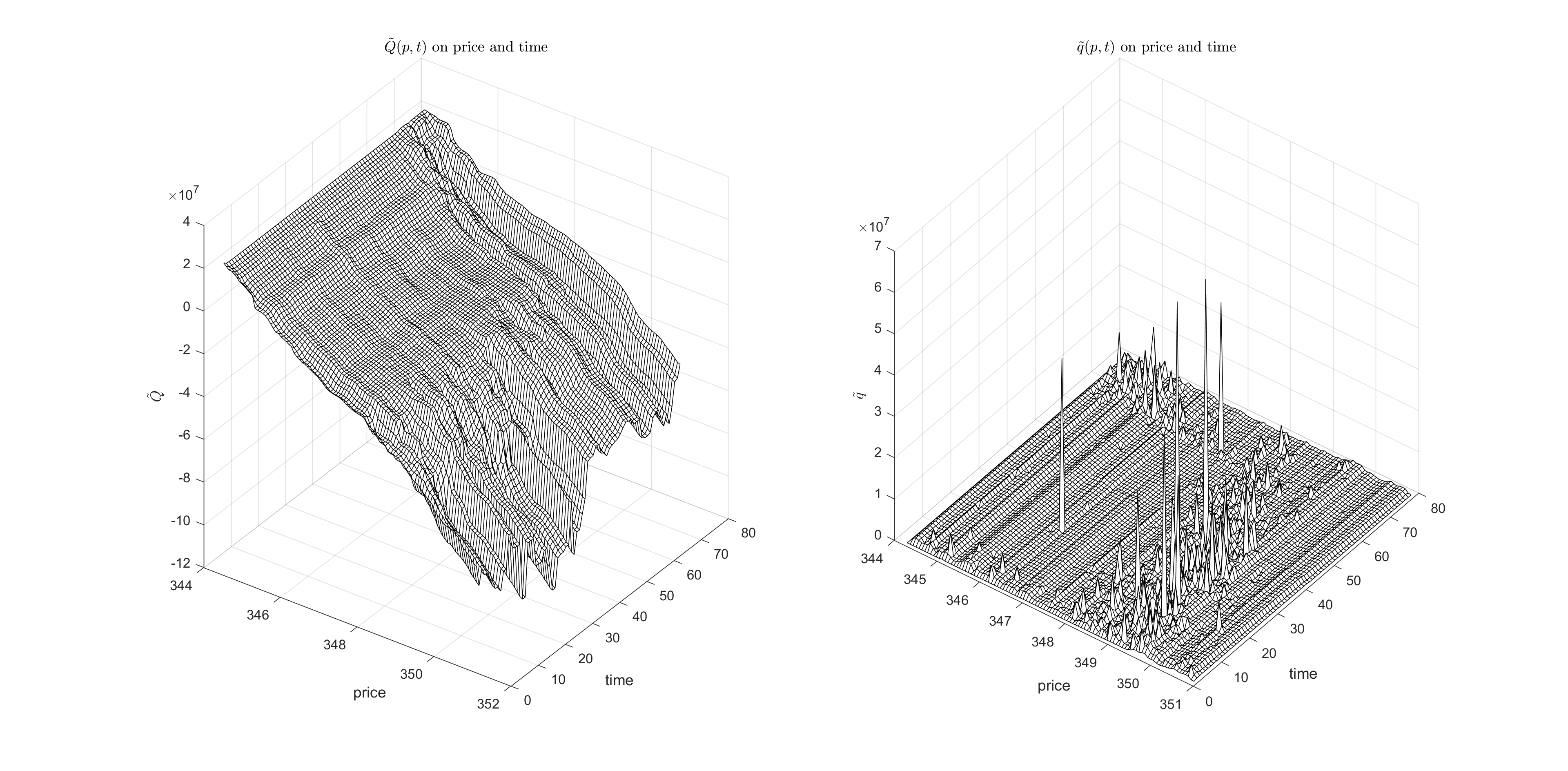} \newline
\caption{\textbf{Simulated data:} The left graph shows the simulated net demand surface of Apple Inc. as of April 1st, 2011 with one random simulation scenario. The price and time ranges of the simulated surface scale the same as the real-time net demand surface. The right graph plots the corresponding net demand increments.}
\label{fig::sim_Q_and_q}
\end{figure}
\end{center}

\subsection{Model Validation}

With the simulated net demand surface, we then calculated (by linear interpolating $Q$ in the price direction) the market clearing price $\pi (t,\omega )$, i.e., the price that solves:
\begin{equation*}
Q(\pi (t,\omega ),t,\omega )=0, \qquad t\in[t_0, T].
\end{equation*}
We simulated $N=$1,000 paths of the price $\pi (.,\omega )$. Let  $c(K)$ and $p(K)$ be the Monte Carlo estimators of the price at time $t <T$ of a call/put option with expiration $T=t_0 +$30 days and strike $K$. By Theorem \ref{th:PCC}:

\begin{eqnarray}
c(K) &=&\frac{1}{N}\sum_{\omega =1}^{N}\max \left( \pi (T,\omega
)-K,0\right) ,  \notag  \label{eqn::avg_price} \\
p(K) &=&\frac{1}{N}\sum_{\omega =1}^{N}\max \left( K-\pi (T,\omega
),0\right).
\end{eqnarray}

The implied volatility for different strike levels $K$ is calculated by
solving

\begin{eqnarray}
\sigma _{call}(K) &=&BS^{-1}(\pi (t_0);K,r,T,c(K)),  \notag
\label{eqn::bs_price} \\
\sigma _{put}(K) &=&BS^{-1}(\pi (t_0);K,r,T,p(K)),
\end{eqnarray}

where $BS^{-1}$ represents the inverse of Black-Scholes European option pricing formula, where we solve the implied volatility level given the clearing price $\pi $. Consistently with our model, we chose the interest rate $r=0.2537\%$, which is linearly interpolated from the market zero-coupon curve (provided in WRDS OptionMetric database). 

\begin{center}
\begin{figure}[ht]
\centering
\includegraphics[scale = 0.3]{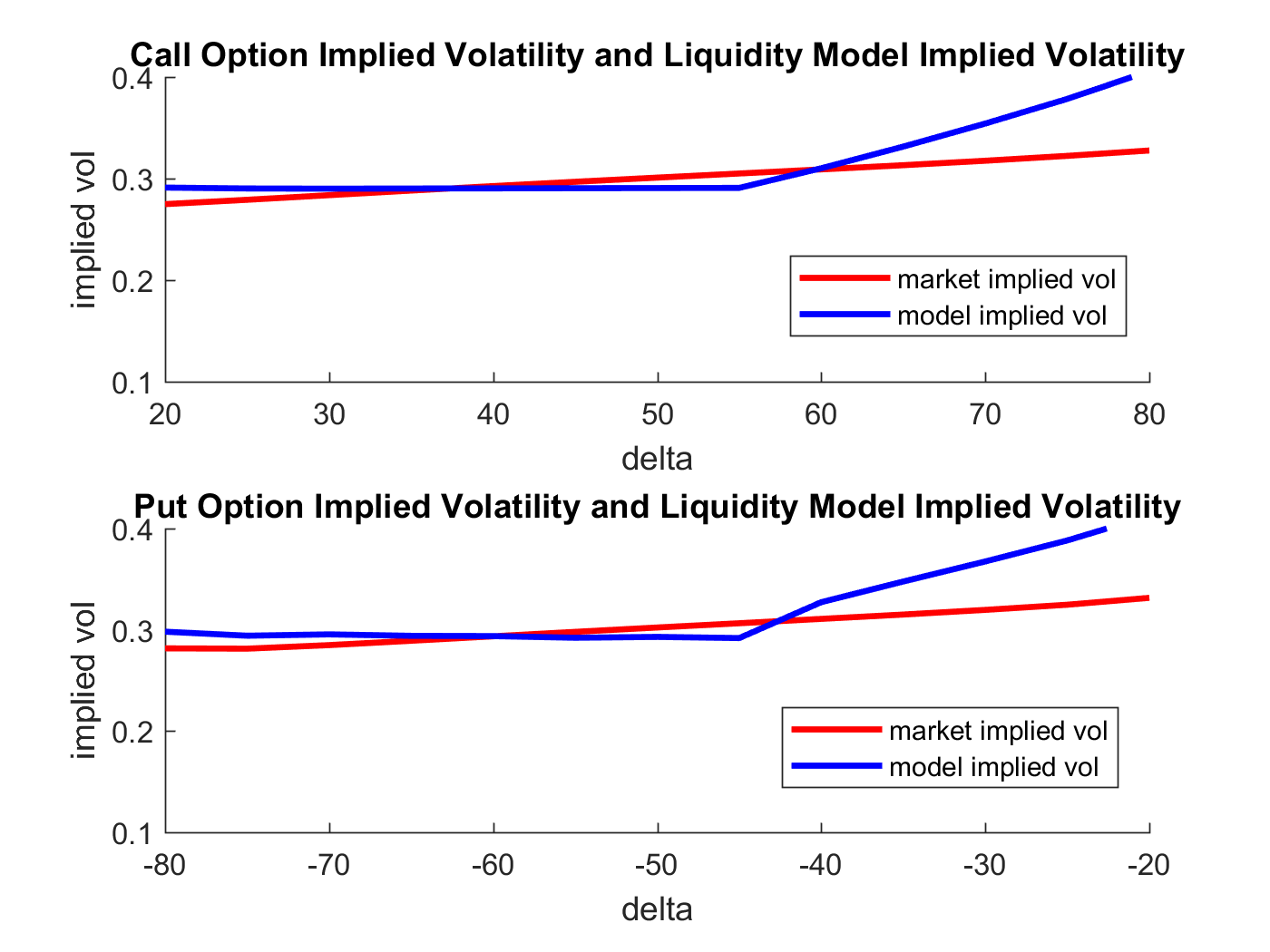} \newline
\caption{The top graph compares the market call option implied volatility
smile and model implied volatility skew. Option maturity is 30 days. The
bottom graphs compares the volatilities implies from put option. The x-axis
represents the delta of corresponding option. The selected stock is Apple
Inc. (NYSE: AAPL) and the market date is Apr. 1st, 2011.}
\label{fig::option_pricing}
\end{figure}
\end{center}

We compared the model implied volatility with the market option implied
volatility with 30-day maturity in Figure \ref{fig::option_pricing}. The
market option implied volatility skew is retrieved from OptionMetric
database. Note that the $x$-axis is indexed by the delta of the
corresponding options, which is the option price sensitivity with respect to
spot level moves. The call option delta is $\Phi (d_{1})$, and the put
option delta is $\Phi (-d_{1})$, where $\Phi (\cdot )$ is the standard
normal cumulative distribution function. The term $d_{1}(t)$ at time $t_0$ is

\begin{equation}
d_{1}=\frac{\log (\pi (t_0)/K)+(r+\frac{\sigma ^{2}}{2})(T-t_0)}{\sigma
\sqrt{T-t_0}}.
\end{equation}

We use the delta instead of strike price in volatility smile representation,
for the reasons that OptionMetric database provides volatility surface on
delta, and that delta is a unified measure of option in-the-moneyness, which
is of more interest to practitioners.

From the above graph, the at-the-money (50\% delta) volatility level from call-type and put-type options matches
almost perfect with the market implied volatility level. The error terms are 0.60\% and 0.21\% (in absolute difference) for call and put option respectively, which shows the pricing model we use is highly market consistent. The model seems to overestimate in-the-money and out-of-money implied volatilities, for both call-type and put-type option. A plausible reason is that simulated paths can hardly reach deep-in-the-moneyness or deep-out-the-moneyness levels. To avoid the volatility smile becoming volatility smirk, limit order book data with multiple trading days would be recommended, but the calibration and simulation approaches remain the same. Another plausible explanation is that there is arbitrage or market incompleteness, so that theorem \ref{th:PCC} does not apply.
\bigskip

\section{Appendix}

\label{sec::App}

\subsection{Proof of Proposition \protect\ref{prop:expo}}

\label{sec::A1}

\begin{proof}
By \eqref{prd},
\begin{equation*}
p(x,t)=p_0(x)+\int_0^t \bar{\mu}_p(x)p(x,r)dr+\int_0^t \int_{\varepsilon}^1%
\bar{\sigma}_p(x,s)B(ds,dr),
\end{equation*}
and therefore the corresponding coefficients in \eqref{P_1} are
\begin{eqnarray*}
\tilde{\sigma}_{P}(x,s,t) &=&\bar{\sigma}_{P}(x_{\min
},s)p(x_{\min},t)+\int_{x_{\min }}^{x}\bar{\sigma}_{p}(y,s)p(y,t)dy, \\
\mu _{P}(x,t) &=&\bar{\mu}_{P}(x_{\min })p(x_{\min },t)+\int_{x_{\min}}^{x}%
\bar{\mu}_{p}(y)p(y,t)dy.
\end{eqnarray*}

We fix $t\in [0,T]$ and construct the solution $\lambda(s,t)$ of %
\eqref{MPR_price} in three steps: first, for $s\in \lbrack
\varepsilon,2\varepsilon )$, then, for $s\in (2\varepsilon ,1]$, and
finally,for $s\in \lbrack 0,\varepsilon )$.

To begin, differentiate both sides of (\ref{MPR_price}) with respect to $x$:
\begin{equation*}
\int_{\varepsilon }^{1}\bar{\sigma}_{p}(x,s)p(x,t)\lambda (s,t)ds=\bar{\mu}%
_{p}(x)p(x,t),\text{ \ \ }x_{\min }\leq x\leq x_{\max },\ \ t\in [0,T].
\end{equation*}

Next, divide both sides by $p(x,t)$:
\begin{equation*}
\int_{\varepsilon }^{1}\bar{\sigma}_{p}(x,s)\lambda (s,t)ds=\bar{\mu}%
_{p}(x), \quad x_{\min }\leq x\leq x_{\max },t\in \lbrack 0,T].
\end{equation*}

Since $\bar{\sigma}_{p}(x,s)=0$ for $s>f(x)$,%
\begin{equation}  \label{MPR-f(x)}
\int_{\varepsilon }^{f(x)}\bar{\sigma}_{p}(x,s)\lambda (s,t)ds=\bar{\mu}%
_{p}(x), \quad x_{\min }\leq x\leq x_{\max },\ t\in [0,T].
\end{equation}

Evaluating this equation at $x=x_{\min }$ leads to a constant value for $%
\lambda (s,t)$ when $s\in \lbrack \varepsilon, 2\varepsilon )$:
\begin{equation*}
\lambda (s,t)=\frac{\bar{\mu}_{p}(x_{\min })}{\int_{\varepsilon
}^{2\varepsilon }\bar{\sigma}_{p}(x_{\min },s)ds}.
\end{equation*}

Next, differentiation of \eqref{MPR-f(x)} with respect to $x$ defines $%
\lambda (s,t)$ for $s\in [2\varepsilon ,1]$:
\begin{equation*}
\lambda \big(f(x),t\big)= \frac{1}{(1-2\varepsilon )\bar{\sigma}_{p}\big(%
x,f(x)\big)}\left(\frac{\partial \bar{\mu}_{p}(x)}{\partial x}%
-\int_{\varepsilon }^{f(x)}\frac{\partial \bar{\sigma}_{p}(x,s)}{\partial x}%
\lambda (s,t)ds\right).
\end{equation*}
recall that $x\mapsto f(x)$ is a bijection from $[x_{\min},x_{\max}]$ to $%
[2\varepsilon, 1].$

Finally, we define $\lambda (s,t)$ for $s\in \lbrack 0,\varepsilon )$. To
this end, we evaluate (\ref{MPR_price}) at $x=x_{\min }$:
\begin{eqnarray*}
\int_{0}^{\varepsilon }\bar{\sigma}_{p}(x_{\min },s)p(x_{\min },t)\lambda
(s,t)ds &=&\bar{\mu}_{p}(x_{\min })p(x_{\min },t) \\
&&-\int_{\varepsilon }^{1}\bar{\sigma}_{p}(x_{\min },s)p(x_{\min },t)\lambda
(s,t)ds,
\end{eqnarray*}
and choose a constant solution:%
\begin{equation*}
\lambda (s,t)=\frac{\bar{\mu}_{p}(x_{\min })-\int_{\varepsilon }^{1}\bar{%
\sigma}_{p}(x_{\min },s)\lambda (s,t)ds}{\varepsilon \int_{0}^{\varepsilon }%
\bar{\sigma}_{p}(x_{\min },s)ds}\text{, \ \ }0\leq s<\varepsilon.
\end{equation*}
\end{proof}

\subsection{Proof of Lemma \protect\ref{lem::model_hold}}

\label{sec::A2}

\begin{proof}
Inequality (\ref{MC_hold}) holds for all $t\in \lbrack 0,T]$ if%
\begin{equation}
Q(0,0) + \delta _{0,\max} - \int_{0}^{S}q(p,0)dp-\delta _{1,\min }S+x_{\max
}\leq 0.
\end{equation}
while inequality (\ref{Bound_hold}) holds for all $t\in \lbrack 0,T]$ if,
for some $0<\varepsilon <1$,
\begin{equation}
Q(0,0)-\int_{0}^{\varepsilon }q(p,0)dp-\delta _{1,\max }\varepsilon S +x_{\min
}\geq 0.
\end{equation}

If we assume that $x_{\max }$, $x_{\min }$ and $S$ are given, we solve these
inequalities for $\delta _{1,\min }$ and $\delta _{1,\max }$:%
\begin{eqnarray*}
\delta _{1,\min } &\geq &\frac{x_{\max }+Q(0,0)+\delta _{0,\max
}-\int_{0}^{S}q(p,0)dp}{S}, \\
\delta _{1,\max } &\leq &\frac{x_{\min }+Q(0,0)-\int_{0}^{\varepsilon
}q(p,0)dp}{\varepsilon S}.
\end{eqnarray*}%
\bigskip

We must first check that $\delta _{1,\min }<\delta _{1,\max }$ , which holds
if:%
\begin{equation*}
\frac{x_{\max }+Q(0,0)+\delta _{0,\max }-\int_{0}^{S}q(p,0)dp}{S}<\frac{%
x_{\min }+Q(0,0)-\int_{0}^{\varepsilon }q(p,0)dp}{\varepsilon S},
\end{equation*}
or
\begin{equation*}
\varepsilon <\frac{x_{\min }+Q(0,0)-\int_{0}^{\varepsilon }q(p,0)dp}{%
x_{\max }+Q(0,0)+\delta _{0,\max }-\int_{0}^{S}q(p,0)dp}.
\end{equation*}
Observe that the right-hand side is strictly positive, so that we can choose $\varepsilon >0$.
The price $\varepsilon S $ is less than $S$ if
\begin{equation*}
\delta _{0,\max }>x_{\min }-x_{\max }+\int_{0}^{S}q(p,0)dp.
\end{equation*}
\end{proof}

\bigskip

\bibliographystyle{apalike}
\bibliography{bib}

\end{document}